%% file: ms.tex
\documentclass[runningheads,a4paper]{llncs}

\usepackage{amssymb}
\usepackage{graphicx}

\usepackage{url}
\urldef{\mailsa}\path|{alfred.hofmann, ursula.barth, ingrid.haas, frank.holzwarth,|
\urldef{\mailsb}\path|anna.kramer, leonie.kunz, christine.reiss, nicole.sator,|
\urldef{\mailsc}\path|erika.siebert-cole, peter.strasser, lncs}@springer.com|    
\newcommand{\keywords}[1]{\par\addvspace\baselineskip
\noindent\keywordname\enspace\ignorespaces#1}

\usepackage[table,dvipsnames]{xcolor}
\usepackage{colortbl}
\usepackage{listings}
\usepackage{multirow}
\usepackage{amsfonts}
\usepackage{amsmath}
\usepackage{algpseudocode}% http://ctan.org/pkg/algorithmicx

\usepackage{algorithm}% http://ctan.org/pkg/algorithm

\usepackage{syntax}
\usepackage{mathtools}
\usepackage{algorithmwh}
\usepackage{multicol}
\usepackage{float}
\usepackage{semantic}
\usepackage{algorithm}
\usepackage{listings}
\usepackage{caption}
\usepackage{framed}
\usepackage{graphics}
\usepackage{tikz}
\usepackage{pgf}
\usetikzlibrary{arrows,automata}
\usepackage[latin1]{inputenc}
\floatstyle{boxed}

\makeatletter
\def\BState{\State\hskip-\ALG@thistlm}
\makeatother

\begin{document}

\mainmatter  % start of an individual contribution

\title{Beyond-Regular Typestate}

\titlerunning{}

 \author{Ashish Mishra%
 \ \ \ \ Y. N. Srikant}

\authorrunning{}
% % (feature abused for this document to repeat the title also on left hand pages)
% 
% % the affiliations are given next; don't give your e-mail address
% % unless you accept that it will be published
 \institute{Indian Institute of Science,\\
 Bangalore, India.\\
 \url{{ashishmishra,srikant}@csa.iisc.ernet.in}}


\toctitle{Lecture Notes in Computer Science}
\tocauthor{Authors' Instructions}
\maketitle

\begin{abstract}
We present an extension for regular typestates, called Beyond-Regular Typestate(BR-Typestate), which is expressive enough to model non-regular properties of programs and protocols over data. We model the BR-Typestate system over a dependently typed, state based, imperative core language, and we prove its soundness and tractability. We have implemented a prototype typechecker for the language, and we show how several important, real world non-regular properties of programs and protocols can be verified.

\end{abstract}

\input{introduction}

\input{dependenttypestate}

\input{analysis}

\input{results}

\input{relatedworks}

\input{conclusion}

\bibliographystyle{plain}
\bibliography{ref}

\input{appendix}

 \end{document}

%% file: introduction.tex
\section{Introduction}
\label{introduction}

To quote Strom and Yemini, the originator of the \textit{Typestate}~\cite{typestate}- ``while \textit{type} of data defines what operations are allowed on data for the life time of the data, \textit{typestate} defines which operations are valid in a given context or state of the data''. Typestates have been a useful concept to model and reason about the \textit{stateful} \textit{effect} systems~\cite{Fink,Ashish} from varied domains. Consider the Buffer \textit{State} (analogous to a \textit{class} in Object Oriented paradigm) in Figure~\ref{listing:simpletype}, with the allowed operations \textit{add, remove} and \textit{print}. Types can enforce what operations are allowed on data. However, since the types associated with a datum is immutable, it can not model program properties such as, \textit{add or remove from the buffer, only if the buffer is in open state.} Typestates associate such mutable types to data objects. Typestate example in Figure~\ref{listing:typestate} defines two sub-typestates of the earlier Buffer state, \textit{OpenBuffer} and \textit{ClosedBuffer}. The open(close) operation transits the ClosedBuffer(OpenBuffer) to open(close) state. Figure~\ref{simplets} shows a regular typestate property automaton for Figure~\ref{listing:typestate}. Normally, these typestate properties are modeled and enforced using types~\cite{typestate,typestates-for-objects}, or could be a feature of the language and enforced statically or at runtime~\cite{TSOP}.

% \begin{figure*}[htbp]
% \begin{minipage}{.3\textwidth}
% \begin{lstlisting}[xleftmargin=1pt, language=Java, basicstyle=\selectfont\ttfamily\scriptsize, numbers=left, extendedchars=true, numberstyle=\tiny, xrightmargin=1pt] 
% state Buffer {
%     
%     var [item];
%     void add();
%     item remove();
%     void print(); 
% }
% \end{lstlisting}
% \caption{A Buffer State}
% \label{listing:simpletype}
% \end{minipage}
% \begin{minipage}{.5\textwidth}
%  
% \end{minipage}
% \end{figure*}

\begin{figure*}
\begin{minipage}{.3\textwidth}
\begin{lstlisting}[language=Java, basicstyle=\selectfont\ttfamily\scriptsize, numbers=none, extendedchars=true, numberstyle=\tiny] 
state Buffer {
    
    var [item];
    void add();
    item remove();
    void print(); 
}
\end{lstlisting}
\caption{A Buffer State}
\label{listing:simpletype}
\end{minipage}
\hspace{.5cm}
\begin{minipage}{.3\textwidth}
\begin{lstlisting}[language=Java, basicstyle=\selectfont\ttfamily\scriptsize, numbers=none, extendedchars=true, numberstyle=\tiny] 
state Buffer {
    
    var [item];
    void add();
    item remove();
    void print(); 
}
state OpenBuffer{
  var [item]
  void add();
  item remove();
  void close();
  void print();
}
state ClosedBuffer{
  var [item];
  void open();
  void print();
}
\end{lstlisting}
\captionsetup{width=.5\textwidth}
\caption{Open and Close Buffer States}
\label{listing:typestate}
\end{minipage}
\begin{minipage}{.3\textwidth}
\begin{tikzpicture}[->,>=stealth',shorten >=1pt,auto,node distance=2cm,
                    semithick]
  \tikzstyle{every state}=[fill=white,draw=black,text=black]

  \node[state] (A)                    {$open$};
  \node[initial,state] 		     (B)  [below of=A]      {$closed$};

  \path (A) edge [loop above] node {$add, remove, print$} (A)
            edge [bend left]	       node {$close$} (B)
    
	(B) edge [loop below] node {$print$} (B)
            edge [bend left] node {$open$} (A);

\end{tikzpicture}
\caption{FSM for the Simple Typestate Property}
\label{simplets}
\end{minipage}
\end{figure*}

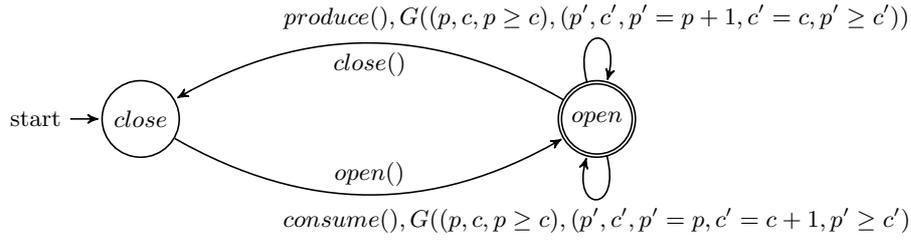
\begin{figure*}[htbp]
\centering
\begin{tikzpicture}[->,>=stealth',shorten >=1pt,auto,node distance=6cm,
                    semithick]
  \tikzstyle{every state}=[fill=white,draw=black,text=black]

  \node[initial,state] (A)                    {$close$};
  \node[state,double] (B) [right of=A]        {$open$};

  \path  (A) edge [bend right] node {$open()$} (B)
  
	  (B) edge [loop above] node {$produce(), G((p, c , p \geq c), (p' , c', p' = p + 1, c' = c, p' \geq c'))$} (B)
	      edge [loop below] node {$consume(), G((p, c , p \geq c) , (p' , c', p' = p,  c' = c + 1, p' \geq c')$} (B)
	      edge [bend right] node {$close()$}									(A);

\end{tikzpicture}
\caption{Counter Machine defining the Invariant property for Producer Consumer over Buffer State Object}
\label{fig:ProdCons-graphical}
\end{figure*}

\begin{figure*}[htbp]
\centering
\begin{tikzpicture}[->,>=stealth',shorten >=1pt,auto,node distance=2cm,
                    semithick, state/.style=state without output]
  \tikzstyle{every state}=[fill=white,draw=black,text=black]

\node[state, initial] (A0) at (-2, 0) [label={$(0,0)$}] {$close$};
\node[state] (A) [label={$(0,0)$}] {$open$};
%\node[state] (B) at (2,0) [label={$(1,0)$}] {$open$};
\node[state] (C) at (2,0) [label={$(2,0)$}]{$open$};
\node[state] (D) at (4,2) [label={$(2,1)$}]{$open$};
\node[state] (E) at (4,0) [label={$(3,0)$}]{$open$};
\node[state] (F) at (6, 0) [label={$(3,0)$}]{$open$};
\node[state] (G) at (6, 2) [label={$(2,2)$}]{$open$};
\node[state, draw=red] (H) at (8,2) [label={$(2,3)$}] {$open$};

  \path (A0) edge [bend left] node {$open()$} (A)
	(A) edge [bend left] node {$p();p()$} (C)
	%(B) edge [bend left] node {$p()$} (C)
	(C) edge 		node {$c()$} (D)
	(D) edge 		node {$c()$} (G)
	(G) edge 		node {$c()$} (H)
	(C) edge [bend left]		node [below] {$p()$} (E)
	(C) edge [bend right]		node [below] {$p()$} (F);

\end{tikzpicture}
\caption{A trace for BR-Typestate change for the Prod-Consumer example, invariant $\forall node, p \geq c$. $p() = produce(), c() = consume()$} 
\label{fig:ts-prodcons-trace}
\end{figure*}
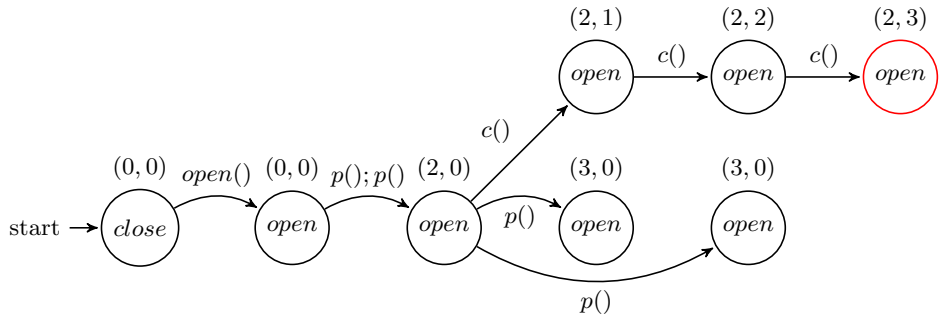

Now, let us consider a slightly richer example of a Buffer object shared between a producer and a consumer process. The buffer provides library methods \textit{produce} and \textit{consume} to these processes. An important runtime property which a producer-consumer model like this must adhere to is- ``At any time during the execution the number of items put into the Buffer must be greater than or equal to the number of items consumed from the Buffer''. At the same time, the items can be produced or consumed only when the Buffer is in \textit{Open} state. 

Figure~\ref{fig:ProdCons-graphical} shows a multiple counter machine~\cite{kozenBook} modeling such a producer consumer problem over a buffer. The machine's states model the states of the Buffer. The number of items produced and consumed are captured using two counters. A transition in the machine is of the form ($\alpha, G(\phi_i, \phi_2)$), where $\alpha$ is an action (like produce or consume) and $G(\phi_1, \phi_2)$, is the guard condition for the transition, requiring $\phi_1$ and guaranteeing $\phi_2$. The property stated above could be defined as an invariant on such a machine ($p \geq c$ in this case). 

The language needed to express and enforce this program property is context-free and thus the regular \textit{Typestate} lacks expressiveness to model such a property~\cite{typestates-for-objects}.

\begin{figure*}[htbp]
\begin{lstlisting}[xleftmargin=1pt, language=Java, basicstyle=\selectfont\ttfamily\scriptsize, numbers=left, extendedchars=true, numberstyle=\tiny, xrightmargin=1pt, mathescape=true] 
state ProducerConsumer {
    type SB :  $\Pi$ ($\phi$(p, c), Buffer); 
    var (1, SB($\phi$(0, 0), OB)) buffer = new OB();   
    void open((1, SB($\phi$(p, q, p >= q), CB)) >> (1, SB($\phi$(p, q, p >= q), OB)) buf)[]{ buf.open();}
    void produce((1, SB($\phi$(p, q, p >= q), OB)) >> (1, SB($\phi$(p + 1, q, p + 1 >= q), OB)))[]{ buf <- (1, SB($\phi$(p + 1, q, p + 1 >= q), OB));}
    void consume((1, SB($\phi$(p, q, p >= q), OB)) >> (1, SB($\phi$(p, q + 1, p >= q + 1), OB)))[]{ buf <- (1, SB($\phi$(p, q + 1, p >= q + 1), OB));}
    void close((1, SB($\phi$(p, q, p >= q), OB)) >> (1, SB($\phi$(p, q, p >= q), CB)))[]{ buf.close();}
}
\end{lstlisting}

\begin{lstlisting}[xleftmargin=1pt, language=Java, basicstyle=\selectfont\ttfamily\scriptsize, numbers=left, extendedchars=true, numberstyle=\tiny, firstnumber=9, xrightmargin=1pt, mathescape=true] 
state Main{	
    void  main()[]{
    var (1, _) pc = new ProducerConsumer();
    var (1, SB($\phi$(0, 0), CB)) buffer = new CB();   
    pc.open(buffer); pc.produce(buffer); pc.produce(buffer);
    match(buffer){
      case OpenBuffer { pc.consume(buffer); pc.consume(buffer); pc.consume(buffer);}
      case ClosedBuffer { pc.produce(buffer);}
      default { pc.produce(buffer); } };}
\end{lstlisting}
%\end{multicols}

\caption{Example Producer-Consumer}
\label{fig:ProdCons-prog}
\end{figure*}

Figure~\ref{fig:ProdCons-prog}, contains the source for a simple Producer Consumer model over a Buffer as described, in our dependently typed language (described later). The State has a Buffer field and a set of methods \textit{open, produce, consume, close}. Each field is annotated with its type which could be a user defined \textit{dependent type}~\cite{dependenttype}, dependent on the runtime values of some dependent term.

Typestates are modeled as instances (line 3) of user defined dependent type families (line 2). Each method has a Hoare style pre and post constraints, which are modeled as a special change type ``\textit{$\gg$}'' that restricts the operations allowed on an object thereby simulating the guarded transitions of the counter machine for the property described earlier. For example, the annotations on method produce in state ProducerConsumer, restricts the production of items to the input Buffer object \textit{buf} only if it is in \textit{open}(OB) state and the number of items produced are greater than or equal to the number of items consumed from buf. 

A typestate in our model is a predicate over object \textit{States} (a regular Typestate)  and an extra set of Presburger formulas. Given these dependently typed annotations with dependent terms coming from a restricted domain, we can mechanically verify that every well typed method and (in turn the whole program) satisfies the annotated pre-condition and guarantees the annotated post-condition. With such an extension, we can model and enforce the guards of multiple counter machines and can enforce these beyond regular program properties {\it with static type checking}, and we call our extension as {\it Beyond-Regular} Typestate (BR-Typesatate). 
There are various languages (both research and real world) which have the full capacity of these dependent types which allow the types to capture and typecheck very complex problems statically. The issue with these languages is that Typechecking for dependent types is undecidable in general (constraint satisfaction is as hard as program equivalence checking)~\cite{Cayenne}, (e.g. Coq, Martin-L\"{o}f type theory(underlying NuPrl) etc.). 

Figure ~\ref{fig:ts-prodcons-trace} shows a property violating trace for the main code fragment. We associate a pair (p, c) representing the number of items produced and consumed respectively till now (shown above the state). Thus, the property checking reduces to the reachability problem for a node with ($p_i, c_i$) as its constraint, such that $p_i < c_i$. The figure shows one such violating trace for the above code with violating node colored red. The violation is caused due to the possible execution of the \textit{OpenBuffer} case (line 15) of the \textit{match} expression.

\subsection{Our Contribution}

\begin{itemize}
\item We present the concept of Beyond-Regular Typestate that has higher expressiveness compared to the regular typestate and can model and verify non-regular program properties.
 
\item We implement this concept as a restricted dependent type system over an imperative dependently-typed core language inspired by ``Typestate-Oriented Programing'', and give the complete formalism for system.

\item We present a formal proof of the correctness and  the decidability of typechecking for our BR-Typestate system. We have also implemented a prototype typechecker for our typestate system.

\item We model several non-regular real world typestate program properties in our language and verify them using the BR-typestate system.

\end{itemize}

The outline of the paper is as follows In section~\ref{dependenttypestate}  we present the formal
language and the BR-Typestate system. In section~\ref{analysis},  we discuss all the important results
and formal properties of our language and the BR-Typestate. Section~\ref{results}, presents some of
the important non-regular program properties and the empirical results that we have generated. Related work and conclusions form the content of sections~\ref{relatedworks} and~\ref{conclusion} respectively.

%% file: dependenttypestate.tex
\section{Beyond-Regular Typestate}
\label{dependenttypestate}
Beyond-Regular(BR) Typestate extends the regular typestate to depend on auxiliary terms. Theoretically, the base terms on which the typestate could depend could be any expression in the language, but this will cause the reasoning over such a system undecidable. Thus, in our work we restrict these base terms to belong to a smaller and less expressive yet decidable domain of Presburger Arithmetic formulas. The expressions in the language might mutate the type-state of the terms. We also restrict these possible mutations so as to make the dependent base terms domain closed under these mutating operations. The utility and the power of these extensions and restrictions will be discussed in detail in section~\ref{analysis}.

\begin{figure*}[t]
\centering
\begin{tikzpicture}[->,>=stealth',shorten >=1pt,auto,node distance=3cm,
                    semithick, state/.style=state without output]
  \tikzstyle{every state}=[fill=white,draw=black,text=black]

\node[state, initial] (A) [label={$(0,0)$}] {$close$};
\node[state] (B) at (3,0) [label={$(0,0)$}] {$open$};
\node[state] (C) at (6,0) [label={$(2,0)$}]{$open$};
\node[state] (D) at (8,2) [label={$(2,1)$}]{$open$};
\node[state] (E) at (10,0) [label={$(p,c)$}]{$open$};

  \path (A) edge [bend left] node {$open()$} (B)
	(B) edge [bend left] node {$(produce())^2$} (C)
	(C) edge 		node {$consume()$} (D)
	(C) edge [bend right, dashed]		node [below] {$(produce() \mid consume())^{*}$} (E)
	(D) edge [dashed]		node {$(produce() \mid consume())*$} (E);

\draw [dotted] (8.5,2.5) -- (10,4);
\draw [dotted] (10,0) -- (12,0);

\end{tikzpicture}
\caption{A possibly infinite Trace for BR-Typestate change for the Producer-Consumer example, invariant being, $\forall \ nodes, p \geq c$}.
\label{fig:tsexample-infinite}
\end{figure*}
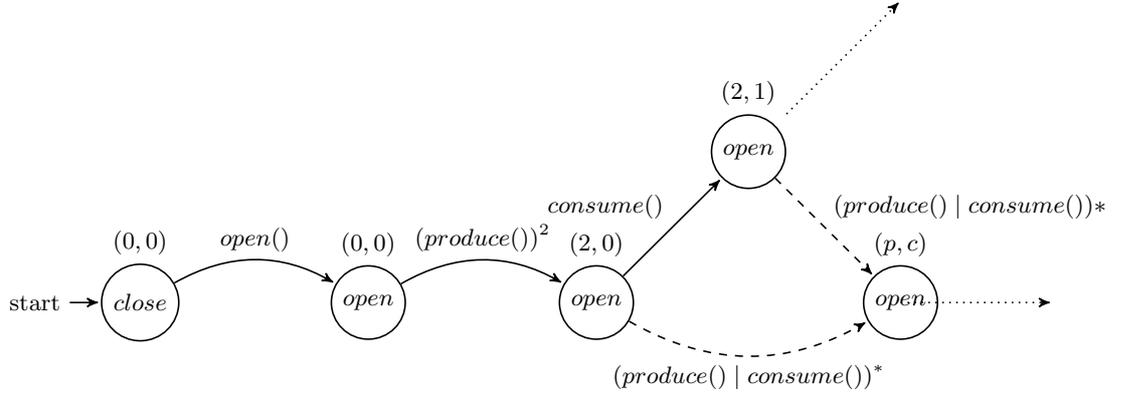

With the intuitive informal understanding of the concept of BR-Typestate, now we present a more formal definition for it-
\begin{definition}[Beyond-Regular Typestate]
A BR-Typestate \textit{BR-ts} for an object \textit{a}, is represented as $a@ts$ and is defined as an instance of a dependent function type family $\Pi_{(\phi : \Phi, s : S)}.\tau$, where $\Phi$ is the type of dependent base terms domain, (Presburger Formulas) in our concrete typestate system and $S$, is the type of the finite state set available in regular typestate. A typestate will be some member of this type family for a given dependent base term and a given state. 

\end{definition}

Thus each node along with the universal invariant $\forall nodes, p \geq c $ in Figure~\ref{fig:tsexample-infinite}, represents a BR-ts for a Buffer object. Thus a given node with a state \textit{open} and (p, c) pair as (c1, c2) represents a BR-ts $[\{(p == c1, c == c2) \wedge c1 \geq c2\}/ \phi, open/S]\tau$.

%%%%%%%%%% REMOVED BY Srikant Sir

\subsection{Core Language}
\label{corelanguage}

\subsubsection{Syntax}
We present a small, core language, inspired by and built upon the ideas from~\cite{TSOP,Plaid,typestates-for-objects}. The language is a state oriented, statically typed imperative programming language with restricted dependent types. The language also has States in place of Classes, along with fields, methods,
and variables. We have highlighted the new features of the language as compared to the earlier typestate oriented programming languages and typestate works in Table~\ref{table:syntax}. The language allows definitions of user defined dependent function type family(\textbf{typefam}), and instantiations of these functions with particular dependent terms(\textbf{type}). These type families and type instantiations let the programmer define types dependent on terms coming from the domain ($\Phi \times states$) and thus allows modeling of BR-Typestates. Moreover, it gives the type system its power to express any possible trace generated by a multiple counter machine (discussed in section~\ref{analysis}). The syntax allows to annotate each method declaration with the Pre and Post BR-Typestate values for parameters and the environment(\textbf{method} in the Table~\ref{table:syntax}). The Pre and Post typestates are represented as typestate transition type \textbf{($\tau_i \gg \tau_j$)}. 
The language requires invariants to be provided explicitly with a {\it while} statement. This assumption is crucial for guaranteeing the termination of the BR-Typestate type-checking since the traces generated by the dependent typesystem are possibly infinite length (the type system can simulate a multiple counter machine).
In section~\ref{analysis} we discuss the automatic inference of such invariants for some particular subclass of program properties.

%%%%%%%%%%%% Replaced with above by Srikant Sir  
% The syntax of the language makes an assumption of explicit invariants annotated with the \textit{while} statement. The assumption is crucial for the termination guarantees of the BR-Typestate typechecking since the traces generated by the dependent typesystem is possibly infinite length(the type system can simulate a multiple counter machine). Thus, an invariant must be associated with the looping construct thereby guaranteeing the termination of the typechecker. Currently, we assume this invariant to be provided by the programmer, later in section~\ref{analysis} we will discuss the automatic inference of such invariants for some particular subclass of program properties. 

Instantiation of \textit{States} using a novel \textit{new} expression, parameterized by a presburger formula(\textbf{new S($\phi : \Phi$)}) is possible. This creates a new object value with the associated BR-typestate parameterized with $(\phi, S)$. Sequential composition is standard as in any imperative language. The static types in the language are either primary types, a state S, a function type (\textbf{$\tau_1 \rightarrow \tau_2$}) or a (permission, type) pair \textbf(a, $\tau$). Besides this, there are special types defining a BR-Typestate instance and its transition. A BR-typestate of a variable, reference or a value is an instance ($\phi, s$).$\tau$ of a dependent function type family $\Pi(\phi : \Phi, s : S).\tau$. The BR-Typestate transition is defined by a typestate transition type($\tau_i \gg \tau_j$) or a method type($\tau_i \rightarrow \tau_2[\tau_i \gg \tau_j]$), which includes a function type and a collection of typestate transition types over parameters and environment variables. Finally, the dependent terms(\textbf{$\Phi$}) of dependent types are either a normal presburger formula or a closed bounded presburger formula. A presburger formula has a standard definition of linear logical constraints over arithmetic addition and constant multiplication terms.

Managing aliases is as imperative in BR-Typestate as is in regular typestate~\cite{TSOP,typestates-for-objects}. To correctly capture the typestate changes in an imperative language, the changes across any possible aliases must be captured. We use the permission system similar to the earlier works on regular typestates which are effective in our current type system as well. There are three permissions, unique (a unique reference to the object) represented by ``1'', shared (atleast two distinct references) represented by  ``2'' and immutable represented by ``-1''. Typing rules for permissions are skipped in view of limited space.

%%%%%%%%%%%%%%%%%%%%%%%%%%%%%%%%%%%%%%%%%%%%%%%%%%%%%%%%%%%%%%%%%%%%
%%%%%%%%%%% New Grammar Rules Based on Typestate for Objects %%%%%%%
%%%%%%%%%%%%%%%%%%%%%%%%%%%%%%%%%%%%%%%%%%%%%%%%%%%%%%%%%%%%%%%%%%%%

\begin{table*}[htbp]
\centering % used for centering table
\begin{tabular}{c c c l} % centered columns (4 columns)
(program) & (P)  & ::= & $state_1, state_2, ... state_n $ in main \\ % inserting body of the table
(state definition) & (state) & ::= & state S case of S \{ $\overline{d} \}$\\
(declaration) & (d) & ::=  & $method \mid field \mid state \mid typefam \mid type$  \\
(method-decl) & (method) & ::= & \colorbox{gray!50}{$\tau_r$ $m_i$ ($\overline{\tau_{ai} \gg \tau_{ai'} a_i}$)[$\overline{\tau_j \gg \tau_{j'} a _j}$] \{ { field; method; stmt; e } \}} \\
(field-decl) & (field) & ::= & (var $\mid$ val) $\tau$ f \\ 
(type-decl) & (type) & ::= & \colorbox{gray!50}{$\gamma$ ($\phi_i, s_i$)} \\
(typeFamily-decl) & (typefam) & ::= & \colorbox{gray!50}{type $\gamma \Pi_{(\phi: \Phi, s : state)}.\tau$} \\
(statement) & (stmt) & ::= & let x = e in stmt \\
		& & &	    $\mid$ let \^{x}.f = e in stmt \\
		& & &	    $\mid$ e $\leftarrow$ e in stmt \\
		& & &	    $\mid$ match (e : S) $\overline{\textnormal{case e \{e\}}}$ \\
 		& & &	    $\mid$\colorbox{gray!50}{while [$\exists.\phi$] ($e_1 : Bool$, $e_2$)} \\
		& & &	    $\mid$ case e \{ e \} \\
(expression) & (e) & ::= & x $\mid$ \^x $\mid$ new S() \colorbox{gray!50}{$\mid$ new S ($\phi : \Phi$)}\\
		& & &	    $\mid$ e.m($e_1, e_2,..., e_p$) \\
		& & &	    $\mid$ \colorbox{gray!50}{e ; e }\\
		& & &	    $\mid$ c \\
(const)     & (c) & ::= & boolliteral $\mid$  intliteral $\mid$ stringliteral \\
(permission) & (a) & ::= & unique (1) $\mid$ shared (2) $\mid$ immutable (-1) \\
(type context) & ($\Gamma$) & ::=  & $\bullet$ $\mid$ $\delta$, $\Gamma$	\\
 & ($\delta$) & ::=  & x : $\tau$ $\mid$  e : $\tau$ $\mid$ d : $\tau$ $\mid$ P : $\tau$ $\mid$ \colorbox{gray!50}{$\tau$ : $\star$}  \\
(heap) & ($\Theta$) & ::=  & $\bullet$ $\mid$ $\theta$, $\Theta$	\\
 & ($\theta$) & ::=  & x, \^x $\mapsto$ value  \\
(value) & value & ::= & c $\mid$ d $\mid$ new S() $\mid$ new S ($\phi: Phi$) $\mid$ $l_i$ \\
(type) & ($\tau$)& ::= &  void $\mid$ int $\mid$ bool $\mid$ string \\
&&&				$\mid$ S \\
(typestate transition)&&&				$\mid$ \colorbox{gray!50}{$\tau_i \gg \tau_j$} \\
(function type)&&&				$\mid$ $\tau_1 \rightarrow \tau_2$ \\ 
(method type)&&&				$\mid$ \colorbox{gray!50}{$\tau_1 \rightarrow \tau_2$ [$\overline{\tau_i \gg \tau_j}$]} \\
&&&				$\mid$ (a, $\tau$) \\
(dependent function type) &&&				$\mid$ \colorbox{gray!50}{$\Pi$ ($\phi : \Phi$, s : S).$\tau$} \\
(Type Family-I )&&&				$\mid$ \colorbox{gray!50}{($\phi$, s).$\tau$} \\

(Dependent Terms Family) & $\Phi$ & ::= & $\phi \mid \lambda_{m_1,m_2,...m_n}.\phi$\\
 & & & \\
(Presburger Formula) & $\phi$ & ::= & b $\mid$ $\phi_1 \wedge \phi_2$ $\mid$ $\phi_1 \vee \phi_2$ $\mid$ $\sim \phi$ $\mid$ $\exists v. \phi$ \\

(Boolean Expression & (b) & ::= & true $\mid$ false $\mid$ i == j $\mid$ $i \leq j$ $\mid$ $i \geq j$ $\mid$ $i \neq  j$ $\mid$ i == int \\

(Arithmetic Expression) & (i) & ::= & c $\mid$ v $\mid$ c * a $\mid$ $i_1$ + $i_2$ $\mid$ - i \\
& & & \\
(variable name) & x , \^{x} this & & \\
(field name) & f & & \\
(method name) & m , main& &  \\
(type family name) & $\gamma$ & &  \\
(state name) & S & & \\ 
(abstract locations) & $l_i$ & & \\
 \hline

\end{tabular}
\caption{Core Language Syntax} % title of Table
\label{table:syntax} % is used to refer this table in the text
\end{table*}

\subsubsection{Operational Semantics of the Core Language}

We present a big step operational semantics for the core-language in the appendix section in the view of limited space. The abstract state of the program is defined as a pair ($\Theta, \Delta$), two variable to value maps mapping reference variables to abstract locations and value variables to values respectively. The big step semantics are presented as judgments $(\Theta, \Delta) \vdash e : \rho; (\Theta', \Delta')$. Such a judgment states that an expression $e$ evaluates in the program state $(\Theta, \Delta)$, to an abstract value $\rho$ and changes the program state to $(\Theta', \Delta')$ in the process. If the expression does not evaluate to a value (like, statements), the judgment drops the returned value $\rho$. Interested readers should refer Appendix, section~\ref{semantics-long} for these semantic rules in Figure~\ref{fig:opsemantics1} and~\ref{fig:opsemantics2} along with their detailed explanation.

\subsubsection{Typing Rules}

\paragraph{\textbf{Type Formation}}
The static dependent type system enforces the type and typestate safety. Figures~\ref{fig:expressiontyping},~\ref{fig:decltyping} and~\ref{fig:subtyping} presents the dependent typing rules for language expressions, well formedness of method, field and state declarations,  and subtyping relations respectively. Figure~\ref{fig:typeformation} presents the standard formation, introduction, computation and other related rules for dependent type family. Each judgment in these rules is of the form \textit{$(\Phi, \Gamma) \vdash e : (\Phi', \tau)$}. It states that in the given typing context $\Gamma$ and dependent base terms constraint environment $\Phi$ (ref. table~\ref{table:syntax}), the expression $e$ is well typed and has a type $\tau$ and the typing of the expression updates the $\Phi$ to $\Phi'$. Any well formed type has a kind which we model as $\star$ in our type system.
Here we discuss in detail only the important and non-standard typing rules in view of limited space, rest are easy to follow.   
The \textbf{T-DepFam-F} rule in Figure~\ref{fig:typeformation}, states that a dependent type family could depend on  a pair ($m$, $s$) of a presburger formula based constraint and a state from the finite state set respectively. The rule states, if $m$ has a well formed type $t$ in the environment and if $s$ has a well formed type S, in the environment extended with ($m : t$), then the type family $\Pi(m : t, s : S).\tau$ is well formed. The type system requires $t$ to be the type of Presburger Arithmetic formula.
The \textbf{T-DepFam-I} and \textbf{T-DepFam-C} are standard introduction and the computation rules for the type family. The next rule \textbf{T-DepFam-C-Eq} defines the rule for equality of two dependent type family instances. It states that two instances of dependent family type are equal iff their dependent base terms are equal component wise. The final rule \textbf{T-Eq} states that if two types are equal as per the tying rules then the type system does not differentiates between them.

\begin{figure*}[htbp]

\begin{center}
\hspace*{5ex} \inference[T-DepFam-F]{ \Phi , \Gamma \vdash t \ type & \Phi \vdash {m : t} & \Gamma, (\Phi, {m : t} )  \vdash S \ type } { \Gamma , \Phi \vdash \Pi({m : t, s : S}).\tau \ type } 

\bigskip

\end{center}

\begin{center}
\hspace*{5ex} \inference[T-DepFam-I]{ (\Phi, \Gamma) \vdash {m : t} & (\Phi, {m : t}), \Gamma \vdash s : S    } { (\Phi, \Gamma ) \vdash \lambda(m : t, s : S).\tau : \Pi({m : t, s : S}).\tau } 

\bigskip

\end{center}

% The substitution rule
\begin{center}
\hspace*{5ex} \inference[T-DepFam-C]{ ( \Phi, \Gamma) \vdash \lambda(m : t, s : S).\tau : \Pi({m : t, s : S}).\tau &  (\Phi, \Gamma \vdash m_c : t) \\ (\Phi, (\Gamma, {m_c : t})) \vdash s_c : S } {  (\Phi, \Gamma) \vdash [m_c / m, s_c / s]\tau : (m_c, s_c).\tau  } 

\bigskip

\end{center}

\begin{center}
\hspace*{5ex} \inference[T-DepFam-C-Eq]{ (\Phi, \Gamma) \vdash m_1 = m_2 : t & (\Phi, \Gamma) \vdash s_1 = s_2 : S } { (\Phi, \Gamma) \vdash (m_1, s_1).\tau  = (m_2, s_2).\tau : \star } 

\bigskip

\end{center}

\begin{center}
\hspace*{5ex} \inference[T-Eq]{ (\Phi, \Gamma) \vdash e : \tau_1 & (\Phi, \Gamma) \vdash \tau_1 = \tau_2 : *} { (\Phi, \Gamma) \vdash e : \tau_2 } 

\bigskip

\end{center}

\caption{Type-Family formation, introduction, computation and equality rules}
\label{fig:typeformation}
\end{figure*}

\paragraph{\textbf{Expression Typing}}

We discuss the most important typing rules. The rule (\textbf{T-new-Dep}) states the typing rule for instantiating a state with initial BR-Typestate. It states, that if the state $S_1$ being instantiated is a well formed declaration(present in State Table, ST), and the presburger formula $\phi_1$ passed as parameter is well formed, then the expression has a dependent type instance $(\phi_1, S).\tau$. The rule also checks the well typedness of the dependent type instance and updates the constraint environment to $\Phi \wedge \phi_1$.
 
The rule (\textbf{T-update}) is the explicit typestate update rule. It first typechecks the right hand expression $e_1$ in the input context and constraint environment and updates the context and the environment. It then checks and updates the type of the left hand expression to the type of the $e_1$. The earlier type of $e$ is discarded, in this sense the Update expression performs a strong type update. 
The rule for match expression (\textbf{T-match}) assigns an arrow type $\tau_1 \rightarrow \tau_u$ to the match expression, where the type of the match conditional expression $e_1$ is $\tau_1$ and $\tau_u$ is a type union over the types for each case expression body. The final constraint environment is a conjunction of the constraints $\phi_i$ imposed by each case expression body $e_i$.  

The (\textbf{T-mcall}) rule  typechecks the base expression $e$ in the pre- context ($\Phi, \Gamma$) and confirms it is an dependent type instance (simple state type $s_i$ can be seen as a constant dependent type $(_, s_i).\tau$). It then  typechecks base expression type, the environment variables type and the actual parameters type against the annotated method type, given by the auxiliary \textit{mtype} routine. Each parameter is checked in a sequentially extended context finally checking the method body $e_m$. The rule ultimately updates the post type of each expression as per the annotated post type in the method type.  

The (\textbf{T-while}) rule checks that the conditional expression $e_1$ is of type \textit{bool} and it updates the incoming environment $\Phi$ to $\Phi_1$, it then validates the associated invariant $\phi$ in $\Phi_1$. It typechecks the body of the while expression while $e_1$ is true ($\Phi_1 \wedge (e_1 == true)$) and confirms whether invariant holds at the end of the while body($\Phi \vDash \exists.\phi$). Finally it validates the invariant when the conditional $e_1$ is false at the exit of the loop. 
%Thus a \textbf{T-while} ensures the satisfaction of the invariant associated before the entry, after some finite iterations and at the exit of the loop.

%%%%%%%%%%%%%%%%%%%%%% Expression Typing %%%%%%%%%%%%%
\begin{figure*}[htbp]

\begin{center}
\hspace*{5ex} \inference[T-var]{ (\Phi, \Gamma) \vdash \tau :: * & (x, \tau) \in \Gamma } { (\Phi, \Gamma) \vdash x : (\Phi, \tau) } 

\bigskip

\end{center}

\begin{center}
\hspace*{5ex} \inference[T-new]{ decl = \textnormal{state S case of Sup \{...\}} & decl \in ST & \tau = (1, S)} {(\Phi, \Gamma) \vdash \textnormal{new S ($\bar{e_1}$)} : (\Phi, \tau)}

\bigskip

\end{center}

\begin{center}
\hspace*{5ex} \inference[T-new-Dep]{ decl = \textnormal{state S_1 case of Sup \{...\}} & decl \in ST & (\Phi, \Gamma) \vdash (\phi_1, S_1).\tau \ type} {(\Phi, \Gamma) \vdash \textnormal{new $S_1$ ($\phi_{1}$)} : ( \Phi \wedge \phi_1 ,(\phi_1, S_1).\tau) }

\bigskip

\end{center}

\begin{center}
\hspace*{5ex} \inference[T-fref]{ (\Phi, \Gamma) \vdash e : (\Phi_1 , \tau_e) & \tau_e = (\phi_e, S_e).\tau 
\\ \textnormal{decl = state $S_e$ case of S \{ $\bar{ts}$ ; $\bar{fs}$ ; $\bar{ms}$ \} }  
\\ decl \in ST & f \in \bar{fs} & ( \Phi_1, (\Gamma, e : \tau_e) \vdash f : ( \Phi_1, \tau)} {(\Phi, \Gamma) \vdash \textnormal{e.f} : (\Phi_1, \tau) }

\bigskip

\end{center}  

\begin{center}
\hspace*{5ex} \inference[T-update]{ (\Phi, \Gamma \vdash e_1 : (\Phi_1, \tau_1) & ( \Phi_1, (\Gamma, e_1 : \tau_1)) \vdash e : ( \Phi_2, \tau_1) } { (\Phi , \Gamma) \vdash \textnormal{e $\leftarrow$ $e_1$} : (\Phi_2, \tau_1)}
 
\bigskip

\end{center}

\begin{center}
\hspace*{5ex} \inference[T-match]{(\Phi, \Gamma) \vdash e_1 : (\Phi_1, \tau_1) & (\Phi_1, (\Gamma, e_1 : \tau_1)) \vdash \overline{e_i : (\Phi_i, \tau_i \rightarrow \tau_{b_i})} \\
 \forall i. \ \tau_i <: \tau_1 & \Phi_u = \bigvee \Phi_i & \tau_u = \bigcup \tau_{b_i}} 
				  {(\Phi, \Gamma) \vdash \textnormal{match $e_1$  $\overline{case \ e_i}$} : (\Phi_u, \tau_1 \rightarrow \tau_u)}

\bigskip

\end{center}

\begin{center}
\hspace*{5ex} \inference[T-let]{(\Phi, \Gamma) \vdash e_1 : ( \Phi_1, \tau_1) & ( \Phi_1, \Gamma , x : \tau_1 , e_1 : \tau_1) \vdash e : (\Phi_2, \tau)} {(\Phi, \Gamma) \vdash \textnormal{ let x = $e_1$ in e} : ( \Phi_2, \tau)}

\bigskip

\end{center}

\begin{center}
\hspace*{5ex} \inference[T-case]{ (\Phi, \Gamma) \vdash e : (\Phi_1, \tau_1) & (\Phi_1, (\Gamma , e_1 : \tau_1) \vdash e_{b} : ( \Phi_2, \tau_b )} 
{(\Phi, \Gamma) \vdash \textnormal{case  $e$  \{  $e_{b}$ \}} : (\Phi_2, \tau_1 \rightarrow \tau_b) }

\bigskip

\end{center}  

\begin{center}
\hspace*{5ex} \inference[T-mcall]{ (\Phi, \Gamma) \vdash e : (\Phi_1, \tau_b) & \tau_b =  (\phi_b, S_b).\tau \\
 mtype(m , S_b) = \textnormal{$T_r$ m($\overline{T_{i} >> T_{i}'} a _i$)[$\overline{T_{this} >> T_{this}'}$]}\{ e_m \} \\
(\Phi_1, (\Gamma, e : \tau_b) \vdash \tau_b <: T_{this} & (\Phi_1, (\Gamma, e : \tau_b) \vdash \overline{e_i: (\Phi_i, \tau_i)} & \overline{\tau_i <: T_i} 
\\ ((\Phi_1 \wedge (\bigwedge_{i} \Phi_i) (\Gamma, e : \tau_b, \overline{e_i : \tau_i)}) \vdash e_m : (\Phi_r, T_r)} 
{(\Phi,\Gamma) \vdash \textnormal{e.m($e_1, e_2, ... e_p$)} : (\Phi_r, T_r)}

\bigskip

\end{center}  

\begin{center}
\hspace*{5ex} \inference[T-while]{ (\Phi, \Gamma) \vdash e_1 : (\Phi_1, bool) & \Phi_1 \vDash \exists.\phi 
\\ (\Phi_1 \wedge (e_1 == true), (\Gamma, e_1 : bool)) \vdash e : (\Phi_2, \tau) & \Phi_2 \vDash \exists.\phi  \\
(\Phi_1 \wedge \Phi_2 \wedge (e_1 == false) \vDash \exists.\phi } 
{(\Phi, \Gamma) \vdash \textnormal{while [$\exists. \phi$] ($e_1$) \{e\}} : (\Phi_2, \tau) }

\bigskip

\end{center}

\caption{BR-Typestate typing rules for expressions}
\label{fig:expressiontyping}
\end{figure*}

%%%%%%%%%%%%%%%%%%%%%%%%% Decl Well formedness %%%%%%%%%%%%%%%%%%%
\paragraph{\textbf{Field, Method and State Well Formedness}}
Figure~\ref{fig:decltyping}, presents the typing rules enforcing and checking the well formedness of  fields, methods and states. Each judgment of the form ($\Phi, \Gamma$) $\vdash$ $d$ : ($\Phi', \star$) states that the declaration $d$ is well formed in the context ($\Phi, \Gamma$) and updates the constraint environment to $\Phi'$. The method declaration rule (\textbf{T-m Decl}) needs some elucidation, it typechecks list of parameters $\overline{e_i}$ against the annotated parameter input types, by sequentially updating the context after each such typecheck. For example it checks $e_1$ in the incoming context against the annotated type $\tau_1$. It then extends the context (both $\Phi$ and $\Gamma$) and further checks the $e_2$ in this extended context. In general it typechecks $e_i$ in the extended context generated by the checking of $e_{i-1}$. Finally, it checks the body of the method declaration in the environment extended by the typechecking of $e_m$. The typechecking of the environment variables, parameters and the body in corresponding contexts implies the well formedness of the method declaration. The rule for state declaration, \textbf{T-s Decl} straight forwardly checks the well formedness of all the types, fields, methods and states declared in the state.

\begin{figure*}[htbp]

% \begin{minipage}{.3\textwidth}
% 
% \hspace*{5ex} \inference[T-f Decl ]{ (\Phi, \Gamma) \vdash \textnormal{$\tau$  type} }
% {(\Phi, \Gamma) \vdash  \textnormal{ $\tau$ f} : (\Phi, *)  }
% 
% \bigskip
% 
% 
% 
% \end{minipage}
% \begin{minipage}{.7\textwidth}
% \begin{center}
% \hspace*{5ex} \inference[T-m Decl ]{ (\Phi, \Gamma) \vdash e_1 : (\Phi_1, \tau_1) \\
% 				      (\Phi_1, (\Gamma, e_1 : \tau_1)) \vdash e_2 : (\Phi_2, \tau_2)
% 				      ... \\
% 				      (\Phi_{m -1}, (\Gamma, e_1 : \tau_1 ... e_{m-1} : \tau_{m-1})) \vdash e_m : (\Phi_m, \tau_m) \\
% 				      (\Phi_{m}, (\Gamma, e_1 : \tau_1 ... e_{m} : \tau_{m}, this : \tau_{this})) \vdash e : (\Phi_m, (\Gamma', \tau_r)) \\
% 				      \forall i. \Gamma'(e_i) = \tau_i' & \Gamma'(this) = \tau_{this}'} 
% {(\Phi, \Gamma) \vdash \textnormal{$\tau_r$ m ($\overline{\tau_i >> \tau_i' e_i}$)[$\tau_{this}$ $>>$ $\tau_{this}'$] \{$e$\}} : (\Phi_m, *)}
% 
% \bigskip
% 
% \end{center} 
% 
% \end{minipage}

\begin{center}
\hspace*{5ex} \inference[T-f Decl ]{ (\Phi, \Gamma) \vdash \textnormal{$\tau$  type} }
{(\Phi, \Gamma) \vdash  \textnormal{ $\tau$ f} : (\Phi, *)  }

\bigskip

\end{center} 

\begin{center}
\hspace*{5ex} \inference[T-m Decl ]{ (\Phi, \Gamma) \vdash e_1 : (\Phi_1, \tau_1) \\
				      (\Phi_1, (\Gamma, e_1 : \tau_1)) \vdash e_2 : (\Phi_2, \tau_2)
				      ... \\
				      (\Phi_{m -1}, (\Gamma, e_1 : \tau_1 ... e_{m-1} : \tau_{m-1})) \vdash e_m : (\Phi_m, \tau_m) \\
				      (\Phi_{m}, (\Gamma, e_1 : \tau_1 ... e_{m} : \tau_{m}, this : \tau_{this})) \vdash e : (\Phi_m, (\Gamma', \tau_r)) \\
				      \forall i. \Gamma'(e_i) = \tau_i' & \Gamma'(this) = \tau_{this}'} 
{(\Phi, \Gamma) \vdash \textnormal{$\tau_r$ m ($\overline{\tau_i >> \tau_i' e_i}$)[$\tau_{this}$ $>>$ $\tau_{this}'$] \{$e$\}} : (\Phi_m, *)}

\bigskip

\end{center}

\begin{center}
\hspace*{5ex} \inference[T-s Decl ]{ \forall f \in fs. (\Phi, \Gamma) \vdash f : (\Phi', * ) \\
\forall t \in tf. (\Phi, \Gamma) \vdash t : (\Phi', * ) \\
\forall m \in ms. (\Phi, \Gamma) \vdash m : (\Phi', * ) \\
(\Phi', \Gamma) \vdash e : (\Phi'' , \tau) } 
{(\Phi, \Gamma) \vdash \textnormal{state S case of S' \{ tf ; fs ; ms ; e \}} : (\Phi'', *)}

\bigskip

\end{center}

\caption{Formation Rules for Field, Method and State Declarations}
\label{fig:decltyping}
\end{figure*}

\paragraph{\textbf{Subtyping}}
Figure~\ref{fig:subtyping}, presents the subtyping rules for the dependent BR-Typestate system. The rule \textbf{T-Sub-Refl} and \textbf{T-Sub-Trans} are standard reflexivity and transitivity rules for subtyping. The rule \textbf{T-Sub-State} defines the subtyping over states, this subtyping relation is definitional in nature such that if sdecl = state S case of $S_1$\{..\}, then $S <: S_1$. The rule \textbf{T-Sub-Str} is the subtyping rule for structural types of the form ($a, \tau$),  $\tau_1$ $<:$ $\tau_2$ holds iff the permission $a_1$ for $\tau_1$ is equal to the permission $a_2$ for $\tau_2$ and recursively ($\tau_{1'} <: \tau_{2'}$).
Rule \textbf{T-Sub-DepTerm} states the subtyping for Dependent term (a presburger formula). It states that if $\phi_1$ and $\phi_2$ are well formed presburger formulas then $\phi_1 <: \phi_2$ iff satisfaction of $\phi_1$ implies the satisfaction of $\phi_2$. Rule \textbf{T-DepFam Sub} defines the subtyping relation for dependent type family instance. It states, the component wise subtyping relation for the dependent type family instance, i.e. if $\phi_1 <: \phi_2$ and $s_1 <: s_2$ then $[\phi_1 / \phi, s_1 / S].\tau  <: [\phi_2 / \phi, s_2 / S].\tau $.
%%%%%%%%%%%%%%%%%%%%%%%%%%%%%%%%%%Subtyping rules %%%%%%%%%%%%%%%%%%%%%%%%
\begin{figure*}[htbp]

\begin{minipage}{.5\textwidth}
\begin{center}
\hspace*{5ex} \inference[T-Sub-Refl]{ \Gamma , \Phi \vdash \tau \ type  } { \tau <: \tau}

\bigskip

\end{center}  
\end{minipage}
\begin{minipage}{.5\textwidth}
\begin{center}
\hspace*{5ex} \inference{ \Gamma , \Phi \vdash \tau_1 <: \tau_2 & \tau_2 <: \tau_3 } {\Gamma , \Phi \vdash \tau_1 <: \tau_3}[T-Sub-Trans]

\bigskip

\end{center}

\end{minipage}

% \begin{center}
% \hspace*{5ex} \inference[T- Sub - Refl]{ \Gamma , \Phi \vdash \tau \ type  } { \tau <: \tau}
% 
% \bigskip
% 
% \end{center}  

% \begin{center}
% \hspace*{5ex} \inference[T- Sub - Trans]{ \Gamma , \Phi \vdash \tau_1 <: \tau_2 & \tau_2 <: \tau_3 } {\Gamma , \Phi \vdash \tau_1 <: \tau_3}
% 
% \bigskip
% 
% \end{center}  
\begin{minipage}{.5\textwidth}

\begin{center}
\hspace*{5ex} \inference[T-Sub-State]{ sdecl \ = \ state \ S \ case \ of \ S_1\{...\} \\ sdecl \in ST } { S <: S_1}
\bigskip

\end{center}  

\end{minipage}
\hspace{1cm}
\begin{minipage}{.5\textwidth}

\begin{center}
\hspace*{5ex} \inference{ \tau_1 = (a_1, \tau_{1'}) & \tau_2 = (a_2, \tau_{2'}) \\ \Gamma , \Phi \vdash a_1  = a_2 & \tau_{1'} <: \tau_{2'}  } { \tau_1 <: \tau_2}[T-Sub-Str]

\bigskip
\end{center}  
\end{minipage}

% 
% \begin{center}
% \hspace*{5ex} \inference[T-Sub-State]{ sdecl \ = \ state \ S \ case \ of \ S_1\{...\} & sdecl \in ST } { S <: S_1}
% 
% \bigskip
% 
% \end{center}  

% \begin{center}
% \hspace*{5ex} \inference[T-Sub-Struct]{ \tau_1 = (a_1, \tau_{1'}) & \tau_2 = (a_2, \tau_{2'}) \\ \Gamma , \Phi \vdash a_1  = a_2 & \tau_{1'} <: \tau_{2'}  } { \tau_1 <: \tau_2}
% 
% \bigskip
% 
% \end{center}  

\begin{minipage}{.5\textwidth}

\begin{center}
\hspace*{5ex} \inference[T-Sub-DepTerm]{ \Phi \vdash \phi_1 \ type , \ \phi_2 \ type \\ \phi_1 \models \phi_2 } { \phi_1 <: \phi_2 }

\bigskip

\end{center}

\end{minipage}
\begin{minipage}{.5\textwidth}
\begin{center}
\hspace*{5ex} \inference{ \Gamma , \Phi \vdash \phi_1 <: \phi_2 \\ \Gamma , \Phi \vdash s_1 <: s_2  } {\Gamma , \Phi \vdash (\phi_1, s_1).\tau <: (\phi_2, s_2).\tau}[T-DepFam Sub]

\bigskip

\end{center}  
\end{minipage}

% 
% \begin{center}
% \hspace*{5ex} \inference[T- Sub - DepTerm]{ \Phi \vdash \phi_1 \ type , \ \phi_2 \ type & \phi_1 \models \phi_2 } { \phi_1 <: \phi_2 }
% 
% \bigskip
% 
% \end{center}  
% 

% \begin{center}
% \hspace*{5ex} \inference[T- DepFam Sub]{ \Gamma , \Phi \vdash \phi_1 <: \phi_2 & \Gamma , \Phi \vdash s_1 <: s_2  }
% %\\ \Gamma , \Phi \vdash \Pi_{\phi,S}[\phi_1 / \phi, s_1 / S].\tau : \tau_1 & \Gamma , \Phi \vdash \Pi_{\phi,S}[\phi_2 / \phi, s_2 / S].\tau : \tau_2 } 
% {\Gamma , \Phi \vdash (\phi_1, s_1).\tau <: (\phi_2, s_2).\tau}
% 
% \bigskip
% 
% \end{center}  

\caption{Subtyping Rules}
\label{fig:subtyping}
\end{figure*}

%%%%%%%%%%%%%%%%%%%%%%%%%%%%%%%% Well Formedness Rules %%%%%%%%%%%%%%%%%%%%%%%%%%%%%%%%%

%% file: analysis.tex
\section{Discussion and Analysis}
\label{analysis}

\subsection{Type Soundness}
We present a soundness proof for our BR-Typestate system. 
%We believe that the proof of soundness presented here also could be scaled down (shedding away the dependent types) to valid proof for type system soundness of ``Typestate for Object''~\cite{typestates-for-objets}. This can be seen as a side-effect of our work.

\begin{theorem}[Progress]
\label{progress-short}
if $\vdash$ t : $\tau$ then either
\begin{itemize}
 \item t is a value. OR
 \item $\exists$ a term t' such that $t \rightarrow t'$.
\end{itemize}

\end{theorem}

\begin{proof}
We prove the above theorem by induction over the derivation of typing rules for the expressions. (refer Appendix, theorem~\ref{soundness} for a detailed proof.)
\end{proof}

\begin{theorem}[Preservation]
\label{preservation-short}
if $\Gamma$, $\Phi$ $\vdash$ t : $\tau$ and  t $\rightarrow$ t', then ($\Gamma'$, $\Phi'$) $\vdash$ t' : $\tau'$ and  ($\Gamma'$, $\Phi'$) $\vdash$ $\tau'$ type.
\end{theorem}
\begin{proof}

The proof is again by the induction on the derivation of $(\Phi, \Gamma) \vdash t : \tau$. 
We present the argument about the preservation for an important subset of cases and for others the argument is similar. At each step of the induction we assume by the induction hypothesis(IH) the preservation holds for the sub-derivations and then to complete the induction argument we prove that the argument hold for the current step.(refer Appendix, section~\ref{soundness}).
\end{proof}

\begin{theorem}[Soundness]
\label{soundness-short}
The typestate system presented in section \ref{dependenttypestate} is sound. Formally, if a term \textbf{t} is a well typed term in our typestate system, then it will never be a stuck term.
\end{theorem}
\begin{proof}
By Theorem \ref{progress-short} and \ref{preservation-short}
\end{proof}

% \subsection{Invariants for While}
% We require the programmer to provide an invariant (a presburger formula) along with the while loop. Clearly, this could be a non-trivial task in many cases to provide the exact invariant for the looping construct. Importantly, availability of such an invariant either from the programmer or using automatic loop invariant synthesis is must for having a decidable type checking over our dependent type system. In this work we take the first approach and put the onus on the programmer to provide the type-checker this invariant, although an intermediate approach of synthesizing such invariants should be possible taking some hints from the programmer rather than the exact invariant as is discussed in ~\cite{Xanadu}.

\subsection{Expressiveness of BR-Typestate}
\label{equivalencetoMCA}
One crucial question to ask is how expressive is the BR-Typestate system defined earlier. We claim that the language of our type system for BR-Typestate(the language generated by the labeled transitions system defined by the dependent type system) although restricted contains all possible traces generated by a multiple counter machine~\cite{kozenBook}. 
%We prove this claim through a reduction of our dependent type system to a labeled transition system $\mathbb{T}_{BR}$ and then showing a simulation of a multiple counter machine using it. We just state the theorem in this section, interested readers can refer the appendix section for a detailed reduction and proof.

\begin{theorem}[BR-Typestate Expressiveness]
The language of the type system for BR-Typestate(the language generated by the labeled transitions system defined by the dependent type system) contains all possible traces generated by a multiple counter machine.
 
\end{theorem}
\begin{proof}
The proof is by reducing our dependent type system to a labeled transition system ($\mathbb{T}_{br}$), modeling a multiple counter machine using another labeled transition system ($\mathbb{T}_{mca}$) and then showing that $\mathbb{T}_{br}$ simulates $\mathbb{T}_{mca}$.(refer Appendix, section~\ref{BR-tsExpressiveness}).

\end{proof}

\subsection{Decidability of Typecheking BR-Typestate}

The typecheking problem for the BR-Typestate is reducible to constraint solving over Presburger Arithmetic formulas. The decidability of the validity problem of Presburger Arithmetic formulas family makes the type checking decidable in our typestate system. 

\begin{theorem}[Reduction to PAF]
\label{reduction-short}
For any general typing relation $(\Phi, \Gamma) \vdash t : (\Phi', \tau)$ in our typestate system, $\exists. \psi \in Presburger Arithmetic Formula$, such that $(\Phi, \Gamma) \vdash t : (\Phi', \tau)$ holds iff $\psi$ is satisfiable.
\end{theorem}
\begin{proof}
The proof is using an inductive argument on the typing derivations of our typestate system. The routine $\psi(\tau)$ defines the presburger formula for $\tau$. (Refer Appendix, section~\ref{decidability}).
\end{proof}

\subsection{Analysis of the Type Inference  Problem}

As described earlier the BR-Typestate system assumes that the {\bf while} syntax is annotated with a loop invariant and we assumed that this is provided by the programmer. This assumption is essential to guarantee  termination of our typechecking algorithm. This could be a hard task for a novice programmer and challenging
even for an experienced programmer. Fortunately, this burden could be placated in certain special subclasses of programs or properties for which the loop invariants could be effectively computed. The loop invariant inference is based on the efficient and decidable verification results~\cite{FlatMCA,FlatnessNotWeakness,Ibarra} for some known subclasses of multiple counter machines, one of which is the {\it Flat Counter Machine}~\cite{FlatMCA}. A multiple counter machine is termed Flat if there is no nested loop in the transition system for the machine. Huber et. al.~\cite{FlatMCA} show that for such machines we can compute a Presburger arithmetic formula representing the  fixpoint for a single loop. Since the invariants needed in our case are presburger formulas, we can plug in this  fixpoint presburger formula for the
loop body in the incoming BR-Typestate at the entry of the loop. For other general class of properties for which such a fixpoint is not effectively computable, we require the programmer to provide an invariant and leave the automatic inference of these invariants for future work.

%% file: results.tex
\section{Applications and Results}
\label{results}

We now discuss some of the practical real world non-regular program properties which we are able to typecheck and enforce through our Typestate system.

%\subsection{DYCK Languages}
DYCK languages are the languages of balanced parentheses. An example string of a DYCK language is ``()(())''. 
\begin{definition}{\textbf{DYCK language}}
Formally, let $\Sigma_{1}$=\{(,)\} be an alphabet consisting of the left and right parentheses. Given word \textit{u} over $\Sigma_{1}$, 
let $D_{1}(u)$ be the number of occurrences of the left parentheses in \textit{u} minus the number of occurrences of the right parentheses in \textit{u}. A word \textup{u} over $\Sigma_{1}$ is said to be a word of well-balanced parentheses, iff
\begin{itemize}
 \item $D_{1}$(u) = 0, and
  \item $D_{1}(v) \geq$ 0 for any prefix \textit{v} of \textit{u}. 
\end{itemize}
\end{definition}

The DYCK language forms the basis of various constructs in programming languages, Internet domain and other fields. For example, markup languages like html, xml, etc., require the programs to be a string of balanced opening and closing elements. Figure~\ref{fig:Dyck} shows a counter machine modeling a DYCK language. The source in our core language captures the states and guards of such machine and skipped due to space limitation.

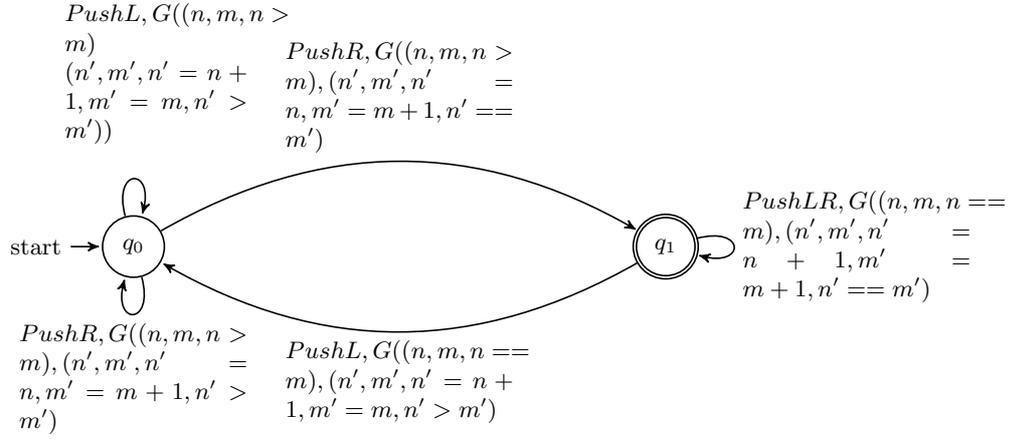
\begin{figure*}[t]
\centering
\begin{tikzpicture}[->,>=stealth',shorten >=1pt,auto,node distance=7cm,
                    semithick]
  \tikzstyle{every state}=[fill=white,draw=black,text=black]

  \node[initial,state] (A)                    {$q_0$};
  \node[state,double] 		(B) [right of=A]  {$q_1$};

  \path (A) edge [loop above] node {\parbox{3cm}{\begin{itemize}
                                                  \item [] $PushL,  G((n, m , n > m)$
						    \item [] $(n' , m', n' = n + 1, m' = m, n' > m'))$ 
                                                 \end{itemize}}} (A)
%$PushL,  G((n, m , n > m), (n' , m', n' = n + 1, m' = m, n' > m'))$}} (A)
            edge [loop below] node {\parbox{3cm}{$PushR,  G((n, m , n > m) , (n' , m', n' = n, m' = m + 1, n' > m')$}} (A)
	    edge [bend left] 	node {\parbox{3cm}{$PushR,  G((n, m , n > m) , (n' , m', n' = n, m' = m + 1, n' == m')$}} (B)
	
	(B) edge [loop right] 	node {\parbox{3cm}{$PushLR, G((n, m , n == m) , (n' , m', n' = n + 1, m' = m + 1, n' == m')$}} (B)
	    edge [bend left] 	node {\parbox{3cm}{$PushL, G((n, m , n == m) , (n' , m', n' = n + 1, m' = m, n'> m')$}} (A);

\end{tikzpicture}
\caption{Counter Machine defining the Invariants property for Dyck Language}
\label{fig:Dyck}
\end{figure*}

%\subsection{Assume guarantee properties}
\begin{definition}{\textbf{Assume Guarantee}}
An important class of program properties which needs to be verified are the \textit{assume-guarantee} properties. These are the properties in which a component (e.g. a function) of the system is specified in terms of the assumptions it makes about its environment (the assume component) and the properties it guarantees about its behavior. The property is naturally represented as $\phi \triangleright \psi$.  

\end{definition}
Assume-guarantee properties are non-regular and hence could not be modeled and enforced using regular typestates. The BR-Typestate by definition models such properties by annotating methods with pre and post constraints. The method assumes certain constraints($\phi$) to be satisfied(\textit{assume}) by the environment and in turn guarantees the output state to satisfy certain constraints(($\psi$), \textit{guarantee}). Thus the \textit{Change Type} $\tau_1 \rightarrow \tau_2 [\overline{\tau_i >> \tau_i'}]$ naturally expresses an assume guarantee property like $\phi \triangleright \psi$, such that $\tau_i \models \phi$ and $\tau_i' \models \psi$.

\begin{definition}{\textbf{Uniform Inevitability Problem}}
The uniform inevitability problem says: there exists some rank \textit{n}, such that every computation sequence of length greater than \textit{n} satisfies some proposition P at rank \textit{n}. The property has been shown to be non-expressible by finite automaton~\cite{Emerson} thus could not be enforced using regular typestate.
\end{definition}

We can model and enforce a variant of Uniform Inevitability problem for a given rank \textit{n} in BR-Typestate. Thus for a given rank \textit{n} and a proposition \textit{P}, we guarantee that a well typed program satisfies -``for all the the paths in the program of length greater than or equal to \textit{n}, the property P holds''.
 
%\subsection{Train speed control algorithm}
\begin{definition}{\textbf{Train speed control algorithm}}
The train speed control algorithm controls the speed of the train and guarantees the collision free running of the trains. A train could be in one of the four states \textit{viz. ontime, braking , late or stopped}. A safety property for such a control system could be defined as - \textit{``the train is never late (or early) by more than 20 seconds''}. The speed control system is regulated via counters keeping track of number of beacons \textit{b} passed on the rails and a global clock ticks \textit{s}, besides this there is another counter which starts in the braking state and counts the ticks during breaking state \textit{d}. Each state is defined as - The train is ontime iff $ s-9 < b < s+9$, its late iff $ b \in [s-9, s-1]$, its early iff $b \geq s+9$ finally, when $b=s+1$, the train is on time again.
\end{definition}
One property of interest to avoid collisions is- $\forall time, \mid b - s \mid \leq 20$,
which could not be enforced using regular typestate. We modeled and enforced this property in our BR-Typestate system. A counter machine for the train speed control protocol is shown in Appendix, Figure~\ref{fig:train-speed-control-fig}. 

Besides the properties described so far in the work, we modeled and enforced a set of other \textit{non-regular} program properties like (1) checking that any path in the program is in language $a^nb^n$.(2) Classic static array bound checking etc.. None of these could be expressed and enforced using regular typestate. 
% Table~\ref{table:results} presents a list interesting \textit{non-regular} properties we could model and enforce using our BR-Typestate system which could not be expressed using regular typestate.
% \begin{table*}[htbp]
% \centering % used for centering table
% \begin{tabular}{c | c | c } % centered columns (4 columns)
% Protocol & Typestate & BR-Typestate \\
% \hline
% $a^nb^n$	& & $\checkmark$ \\
% dyck language	& & $\checkmark$ \\
% Producer-consumer & & $\checkmark$ \\
% Assume-Guarantee & & $\checkmark$ \\
% Train Algorithm	 & & $\checkmark$ \\
% Red-Black Tree & &  \\
% AVL Tree	& &  \\
% Array Bound Check & & $\checkmark$ \\
% Emmerson's uniform inevitably & &  \\
% 
% 
% 
%  \hline
% \end{tabular}
% \caption{Properties Type checked} % title of Table
% \label{table:results} % is used to refer this table in the text
% \end{table*}

%% file: relatedworks.tex
\section{Related Work}
\label{relatedworks}

Our core-language is inspired by and built-upon the Typestate Oriented Programming languages works ~\cite{TSOP,Plaid} but, the BR-Typestate has a static type system over the core language rather than enforcing the typestate in the language and we use a dependent type system to implement it. Modular typestate for object-oriented programs~\cite{typestates-for-objects} models the typestates as predicates over object and handle the issues related to subclasses. This handles regular typestate only. We leave modular BR-Typestate for future research wok. 
Extended Static Checking (ESC) for Java~\cite{ESC} is based on first order logic and general theorem proving. Although ESC is expressive, it does not provide or aim for the decidability and the soundness properties of their static checking, while we show our BR-Typestate system to be sound and our static dependent typechecking to be decidable. 
The domain of dependently typed extensions for languages~\cite{DMLC,Xanadu,liquid,constraintedoo} is also related. These works are some restricted form of dependent types, but our work with a Presburger arithmetic domain as constraint and a core state oriented, imperative language differs from these. The idea of restricting the domain for dependent terms follows from Xi et. al.~\cite{DMLC,Xanadu}, but unlike them we use a decidable class of Presburger formulas for which the exact  typechecking and subtyping is decidable and even inferable in certain cases. Liquid types~\cite{liquid} and other refinement types associate invariants about the runtime values with the data using dependent types and statically verify these invariants. Their emphasis is primarily on the automatic inference of these invariants, compared to these, we focus on increasing the expressiveness of regular typestates, yet keeping the exact typechecking decidable by choosing a decidable logic family as dependent terms. Moreover, while they take a conservative approach of subtyping by embedding the implications of their subtyping rules into a decidable logic, we restrict the dependent terms themselves to a decidable logic fragment there by making the exact typechecking problem decidable. Nathaniel et. al.~\cite{constraintedoo}present a constrained type for an immutable state of a Class, and this work is strictly less expressive than our work where we can model and typecheck invariants on any data of the program.

%% file: conclusion.tex
\section{Conclusion}
\label{conclusion}
We have tried to overcome the expressive limitations of \textit{regular typestate}, by defining the concept of BR-Typestate which is expressive yet decidable. We implemented a restricted dependent type system over a state based, imperative core language. We proved important soundness and decidability results for BR-Typestate and corroborated its effectiveness by verifying several real world non-regular properties.

%% file: appendix.tex
\section{Appendix}
\label{appendix}

\subsection{Operational Semantics of the Core Language}
\label{semantics-long}
 The abstract state of the program is defined as a pair ($\Theta, \Delta$), two variable to value maps mapping reference variables to abstract locations and value variables to values respectively. The big step semantics are presented as $(\Theta, \Delta) \vdash e : \rho; (\Theta', \Delta')$. Such a judgment states that an expression $e$ evaluates in the program state $(\Theta, \Delta)$, to an abstract value $\rho$ and changes the program state to $(\Theta', \Delta')$ in the process. If the expression does not evaluate to a value (statements), the judgment removes the returned value $\rho$. Figure~\ref{fig:opsemantics1} presents these semantic rules for the language. Some of these judgments are self explanatory, while the most interesting ones, most closely relevant to the typestate and BR-Typestate are given by the rules \textbf{mcall, let, match, update, and while}. \textbf{mcall} has a call by value semantics. It checks that the receiver reference is mapped to a non-null (null is a special location) location and then creates an extended program state mapping each formal parameter expression $e_i$ to the values of the corresponding actual parameters and it then evaluates the body of the called method in this new extended state to change the state to $(\Theta_{out}, \Delta_{out})$. The \textbf{match} expression  evaluates the match expression $e$ and further evaluates each of the case expressions $e_i$ in this new program state returning $\rho_{ei}$, and possibly changing the state to ($\Theta_i, \Delta_i$). Since, the match expression could match to any of the possible case expression, we create an over-approximate value for state of the system post completion of the rule. Thus $(\Theta_{out}, \Delta_{out})$, is a union over all the state maps generated by each of the case expressions. The returned value $\oplus \rho_{ei}$ is one of the any possible returned value, thus this can bee seen as an indexed set of values, indexed over the case expression $e_i$. The \textbf{update} rule, refer Figure~\ref{fig:opsemantics2} evaluates the source expression $e'$ of the update expression, changing the state to $(\Theta', \Delta')$ and updates the fields of the target expression $e$, \{ $f_1, ... f_p$ \} by the values of the corresponding fields from the source expression $e'$. The final state is the new updated state with updated maps for each field of $e$ and the $e$ itself. The \textbf{while} rule semantics depend on the value of the conditional expression $b$, if the the condition evaluates to false (\textbf{while-false}) while updating the state to $\Theta', \Delta'$ during evaluation of the $b$, the expression evaluates the next expression (or statement) $e_n$ after the while body. The case for true condition (\textbf{while-true}) is much complex,which evaluates the body of the while statement $e$, in the updated environment and evaluates the next expression $e_n$, only in the new state ($\Theta', \Delta'$), which is obtained after a fix point for the loop is reached.

\begin{figure*}[htbp]

\begin{center}
\hspace*{5ex} \inference[const]{\Delta' = \Delta, (\rho \mapsto c) }{(\Theta; \Delta) \vdash c : \rho ; (\Theta, \Delta')}

\end{center}

\bigskip

\begin{center}
\hspace*{5ex} \inference[val-var]{\Delta' = \Delta, (x \mapsto \rho) & \rho = default(\Gamma(x)) }{(\Theta; \Delta) \vdash x : \rho ; (\Theta, \Delta')}

\end{center}

\bigskip

\begin{center}
\hspace*{5ex} \inference[ref-var]{\Theta' = \Theta, (\textnormal{\^x} \mapsto \rho) & \rho = default(\Gamma(\textnormal{\^x})) }{(\Theta; \Delta) \vdash \textnormal{\^x} : \rho ; (\Theta', \Delta)}

\end{center}

\bigskip

\begin{center}
\hspace*{5ex} \inference[de-ref]{ \Theta(\textnormal{\^x}) = \rho \\ \Theta(\rho) \neq null \\ \Theta(\rho) = new S (a_1 : \rho_1, ...,a_p : \rho_p) \\ \Theta' = \Theta[\textnormal{\^x}.f_j \mapsto \rho_j], (\textnormal{\^x} \mapsto \rho))) & \rho_j = \Theta(a_j) }{(\Theta; \Delta) \vdash \textnormal{\^x}.f_j : \rho_j ; (\Theta', \Delta)}

\end{center}

\bigskip

\begin{center}
\hspace*{5ex} \inference[new]{ \Theta' = \Theta, (\rho \mapsto new S() \mid new S (\phi))}{(\Theta; \Delta) \vdash new S() \mid new S (\phi) : \rho ; (\Theta', \Delta)}

\end{center}

\bigskip

% 
% 
% \begin{center}
% \hspace*{5ex} \inference[assign]{ (\Theta, \Delta) \vdash e : \rho_e; (\Theta', \Delta') \\ \Theta'' = \Theta'[x \mapsto \rho_e]}{(\Theta; \Delta) \vdash x = e ; (\Theta'', \Delta')}
% 
% \end{center}
% 
% \bigskip
% 

\begin{center}
\hspace*{5ex} \inference[mcall]{\Theta(y) \neq null & \Theta(y) = \rho_m \\ \rho_m := \tau_r m (e_1, e_2,...e_p)[]\{ e_b \} \\ \Theta' = \Theta[e_i \mapsto \Theta(f_i)] \\
				   (\Theta' , \Delta \vdash e_b : \rho_b ;(\Theta'', \Delta'') )}{(\Theta, \Delta) \vdash y.m(f_1, f_2,...f_p) : \rho_b ; (\Theta'', \Delta'')}

\end{center}

\bigskip

\begin{center}
\hspace*{5ex} \inference[let]{ (\Theta, \Delta) \vdash e_1 : \rho_{e1}; (\Theta', \Delta') \\ \Theta'' = \Theta'[x \mapsto \rho_{e1}] \\ (\Theta'' , \Delta') \vdash stmt : \rho ; (\Theta_{out}, \Delta_{out})}{(\Theta; \Delta) \vdash \textnormal{let x = $e_1$ in stmt} : \rho; (\Theta_{out}, \Delta_{out})}

\end{center}

\bigskip

\begin{center}
\hspace*{5ex} \inference[match]{ (\Theta, \Delta) \vdash e : \rho_e; (\Theta', \Delta') \\ 
				    (\Theta', \Delta') \vdash e_i : \rho_{ei}; (\Theta_i, \Delta_i)\\
				    \Theta_{out} = \bigcup \Theta_i & \Delta_{out} = \bigcup \Delta_i \\
				    \rho_{out} \oplus \rho_{ei}}{(\Theta; \Delta) \vdash \textnormal{match e case $e_1 \{ b_1\}... e_p \{ b_p\}$} : \rho_{out} ; (\Theta_{out}, \Delta_{out})}

\end{center}

\bigskip

\caption{Big step operational semantics for the core language}
\label{fig:opsemantics1}
\end{figure*}

\begin{figure*}[htbp]

\begin{center}
\hspace*{5ex} \inference[update]{(\Theta, \Delta) \vdash e' : \rho_{e'}; (\Theta' , \Delta') \\
				    \rho_{e'} = new S_t ( f_1 : \rho_{t1}, f_2 : \rho_{t2}, ... f_p : \rho_{tp})\\
				      \Theta'(e) = \rho_e = new S_s ( f_1 : \rho_{s1}, f_2 : \rho_{s2}, ... f_p : \rho_{sp}) \\
					\Theta'' = \Theta [e \mapsto \rho_{e'}] \\
					  \Theta_{out} = \Theta'' [\forall \rho_e.f_i \mapsto \rho_{ti}]\\
					    \Delta_{out} = \Delta'}{(\Theta, \Delta \vdash e \leftarrow e' : \rho_{e'} ; (\Theta_{out}, \Delta_{out}))}

\end{center}

\bigskip
\begin{center}
\hspace*{5ex} \inference[while-false]{ (\Theta, \Delta) \vdash b : false; (\Theta', \Delta') \\ (\Theta', \Delta') \vdash e_n : \rho_n ; (\Theta'', \Delta'')}{(\Theta; \Delta) \vdash \textnormal{while [$\exists.\phi$] b \{ e \}} ; e_n : \rho_n ; (\Theta'', \Delta'')}

\end{center}

\bigskip

\begin{center}
\hspace*{5ex} \inference[while-true]{ (\Theta, \Delta) \vdash b : true; (\Theta', \Delta') \\ (\Theta', \Delta') \vdash e : \rho_e ; (\Theta'', \Delta'') \\ 
					(\Theta'', \Delta'') \vdash e_n : \rho_n ; (\Theta_{out}, \Delta_{out})}{(\Theta; \Delta) \vdash \textnormal{while [$\exists.phi$] b \{ e \}} ; e_n : \rho_n ; (\Theta_{out}, \Delta_{out})}

\end{center}

\bigskip

\begin{center}
\hspace*{5ex} \inference[seq]{ (\Theta, \Delta) \vdash e_1 : \rho_1 ; (\Theta', \Delta') \\ (\Theta', \Delta') \vdash e_2 : \rho_2 ; (\Theta'', \Delta'')}{(\Theta; \Delta) \vdash e_1 ; e_2  : \rho_2 ; (\Theta'', \Delta'')}

\end{center}

\bigskip
% 
% \begin{center}
% \hspace*{1ex} \inference[R-UpdateType]{ }{ }
% 
% \end{center}

\caption{Big step Operational semantics for the core language}
\label{fig:opsemantics2}
\end{figure*}
%%%%%%%%%%%%%%%%% proof of Expressiveness %%%%%%%%%%%%%%%%%%%%%%%
\subsection{Expressiveness of the BR-Typestate type system}
\label{BR-tsExpressiveness}

%\subsubsection{BR-Typestate as an LTS}

\begin{definition}[Labeled Transition System]
A labeled transition system $\mathbb{T}$ over alphabet $\Sigma$ is defined as a tuple $\langle S, A , \rightarrow \pi, F \rangle$, where $S$ is a possibly infinite but countable set of states, $F \subseteq S$ is a set of final states, $\rightarrow \subseteq (S \times A \times S)$ is a transition relation over states on action set $A$ and $\pi : S \mapsto \Sigma$ is a labeling function from states to the alphabet set.

\end{definition}

\begin{definition}[BR-Typestate LTS]
We construct an LTS $\mathbb{T}_{br}$ := $\langle S_{br}, A_{br}, \rightarrow_{br}, \pi_{br}, F_{br} \rangle$ such that-
\begin{itemize}
 \item $S_{br} \subseteq (\Phi \times PS)$, where $\Phi$ represents a Presburger Formulas in the dependent type system while the PS is finite or infinite set of property states, given as dependent terms in our type system. Thus in a set theoretic sense a state conceptually is equal to a dependent type instance in our type system dependent on $\phi \in \Phi, s \in PS$.
\item $A_{br}$ is the set of actions which is the set of transition over the types. The types $\tau_1 \rightarrow \tau_2$, $\tau_i >> \tau_j$ and $\tau_1 \rightarrow \tau_2 [\tau_i \gg \tau_j]$ form the action set for $\mathbb{T}_{br}$. Note that these typing rules only allow presburger arithmetic transitions.
\item The labeling function $\pi_{br}$ is trivial and returns the formula $\phi$ and state s for a given state.

\item The transition relation $\rightarrow_{br}$ - For a given state defined by ($\phi_1, s_1$) and a given action $a \in A_{br}$ is defined as- 
\begin{itemize}
 \item if a = $\tau_i >> \tau_j$ or $\tau_i \rightarrow \tau_j$, with $\tau_i := (\phi_i, s_i).\tau, \tau_j := (\phi_j, s_j).\tau$  then (($\phi_i, s_i$), $(\tau_i >> \tau_j)$, $(\phi_j, s_j)$) $\in \rightarrow_{br}$. 
  \item if a = $\tau_1 \rightarrow \tau_2 [\tau_i \gg \tau_j]$, with $\tau_i := (\phi_i, s_i).\tau, \tau_j := (\phi_j, s_j).\tau$ and $\tau_1 := (\phi_1, s_1).\tau, \tau_2 := (\phi_2, s_2).\tau$  then (($\phi_i, s_i$), $(\tau_i >> \tau_j)$ $(\phi_j, s_j)$) $\in \rightarrow_{br}$ and (($\phi_1, s_1$), $(\tau_1 \rightarrow \tau_2)$ $(\phi_2, s_2)$) $\in \rightarrow_{br}$.
\end{itemize}

\end{itemize}

\end{definition}

We first define a multiple counters automata formally and then present an LTS for such a system. Finally we present a formal proof for $\mathbb{T}_{br}$ simulating the LTS for this Multiple Counters Automata.
\begin{definition}[Multiple Counters Automata]
A multiple counters automata is a tuple $(Q, q_i, C , \delta \subseteq Q \times G(C, C') \times Q)$ where-
\begin{itemize}
 \item Q is a finite set of states.
  \item $q_i \in Q$ is an initial state
  \item $C$ is the finite set of counter variable names, $C'$ is the set of primed counter variable names.
\item $G(C, C')$ is the set of guards built on the alphabets $C, C'$. A member of $G(C, C')$ is a conjunction of atomic formulas of the forms $x \sharp y + c, x \sharp c$, where $x, y \in C \cup C', \sharp \in \{ \geq, \leq, =, >, < \}$ and $c \in \mathbb{Z}.$ 
\end{itemize}

\end{definition}

\begin{definition}[Multiple Counters Automata LTS]
We construct an LTS $\mathbb{T}_{mca}$ := $\langle S_{mca}, A_{mca}, \rightarrow_{mca}, \pi_{mca}, F_{mca} \rangle$ such that -
\begin{itemize}
 \item $S_{mca} \subseteq  Q \times (C \cup C')$, such that if $(q, c_i, c_i', q') \in \delta$, then (q, $c_i$) $\in S_{mca}$ and (q', $c_i$) $\in S_{mca}$.
  \item $A_{mca} \subseteq (C \cup C')$. This defines set of formulas from ($C \cup C'$), which encode the actions of the LTS.
  \item $\rightarrow_{mca} \subseteq (S_{mca} \times A_{mca} \times S_{mca})$.
   \item $\pi_{mca} : S_{mca} \mapsto (C \cup C')$, such that $\forall s_i \in S_{mca} = (q_i, c_i), \pi_{mca}(s_i) = (q_i, c_i)$
    \item $F_{mca} \subseteq S_{mca}$
\end{itemize}
\label{def:def-LTS-mca}
\end{definition}
%%%%%%%%%%% Definition Simulation between two LTS %%%%%%%%%%%%%%
\begin{definition}[Simulation]
Given two LTS $TS_1$ := $\langle S_1, A, \rightarrow_{1}, \pi_{1}, F_{1} \rangle$ and $TS_2$ := $\langle S_2, A, \rightarrow_{2}, \pi_{2}, F_{2} \rangle$. A relation $R \subseteq (S_1 \times S_2)$ is a simulation if $\forall, (p, q) \in R$ and $a \in A$ following holds-
\begin{itemize}
 \item iff $q \in F_2$ then $p \in F_1$. and
  \item iif $(q, a , q') \in  \rightarrow_2$ then $\exists. p' \in S_1$, such that $(p, a, p') \in \rightarrow_1$. and
  \item $(p', q') \in R$.
\end{itemize}
 
If $(p, q) \in R$ then we say that state $p$ simulates state $q$.

\end{definition}
\begin{definition}[Simulation between LTS]
\label{simulationLTS}
Let $p_0$ and $q_0$ be start states for two LTS $T_1$ and $T_2$ respectively. $T_1$ simulates $T_2$ iff $(p_0, q_0) \in R$, where $R \subseteq (S_1 \times S_2)$ is a simulation relation as defined above.
 
\end{definition}

Using these definition now we state and prove important simulation property regarding $\mathbb{T}_{br}$ and $\mathbb{T}_{mca}$.

\begin{theorem}

 If $\mathbb{T}_{br}$ is an LTS for the BR-Typestate type system and another LTS $\mathbb{T}_{mca}$ for the Multiple Counters Automata, then $\mathbb{T}_{br}$ simulates the LTS $\mathbb{T}_{mca}$. Formally. $\exists Sim$. $Sim \subseteq (S_{br} \times S_{mca})$ and start states $p_{0}$ and $q_0$ of $\mathbb{T}_{br}$ and $\mathbb{T}_{mca}$ respectively, then $(p_{0}, q_0) \in Sim $. 
\end{theorem}

\begin{proof}
The proof is an inductive constructive proof on transition relation over $\mathbb{T}_{br}$ and $\mathbb{T}_{mca}$ over finite action set. 

\textbf{Base case} -If $(q_0) = (s_0, c_0) \in F_{mca}$ then by construction we have a state $p_0 \in S_{br}$, such that $p_0 = (c_0, s_0)$ and $p_0 \in F_{br}$.

\textbf{Induction Hypothesis} - Let, for any state $q_{i-2} = (s_{i-2}, c_{i-2}) \in S_{mca}$, then $\exists p_{i-2} = (c_{i-2}, s_{i-2}) \in S_{br}$ such that $(p_{i-2}, q_{i-2}) \in Sim$.

\textbf{Inductive Step}- By IH, $(p_{i-2}, q_{i-2}) \in Sim$, thus by the definition of simulation, states $(p_{i-1}, q_{i-1})$ reachable from $(p_{i-2}, q_{i-2}) \in Sim$. Thus we look at the transitions from $p_{i-1}$ and $q_{i-1}$. 
$\forall$ transitions $\alpha_{mca}$, from $q_{i-1}$, where $\alpha_{mca} := (q_{i-1}, (c_{i-1}, c'{i-1}), q_i) \in \rightarrow_{mca}$ we can always construct a transition $\alpha_{br} := (p_{i-1}, a_i, p_i) \in \rightarrow_{br}$, where $a_i = \tau_{i-1} >> \tau_i$ such that $\tau_{i-1} = (s_{i-1}, c_{i-1}).\tau$ and $\tau_{i} = (s_{i}, c'_{i-1}).\tau$. Thus $(p_{i-1}, q_{i-1}) \in Sim$. Hence by induction, $\forall q_i \in S_{mca}, \exists p_i \in S_{br}$ such that $(p_i, q_i) \in Sim$.

\end{proof}

\begin{corollary}
 $(p_0, q_0) \in Sim$ and thus by definition~\ref{simulationLTS} $\mathbb{T}_{br}$  simulates $\mathbb{T}_{mca}$.  
\end{corollary}
%%%%%%%%%%%%%%%%%%%%%%Proof Of Soundness %%%%%%%%%%%%%%%%%%%%%%%
\subsection{Proof of Soundness of Type System}
\label{soundness}
\begin{theorem}[Progress]
\label{progress}
if $\vdash$ t : $\tau$ then either
\begin{itemize}
 \item t is a value. OR
 \item $\exists$ a term t' such that $t \rightarrow t'$.
\end{itemize}

\end{theorem}
We prove the above theorem by induction over the derivation of typing rules for the expressions.
\begin{proof}
The base cases exists for terms which are values, \textit{viz.} T-New, T-New-Dep and T-mDecl. The case T-Var is trivially satisfied as the term is not typable in an empty context. The interesting cases to consider are T-Let, T-Fref, T-Update, T-Match, T-Case  and T-While. 
\begin{itemize}
 \item T-Let - t := let x = $e_1$ in e. By IH either $e_1$ is a value in which case t reduces to the substitution [value($e_1$) / x]e, or $e_1 \rightarrow e_{1'}$ in which case t $\rightarrow$ t', such that t' := let x = $e_{1'}$ in e.

 \item T-Fref - t := let \^x.f = $e_1$ in e. The argument for the T-Let holds in this case too.

 \item T-Update - t := e $\leftarrow$ $e_1$ ; $e_n$, By IH either $e_1$ is a value, in which case t is reduced to [value($e_1$) / e]$e_n$, or $e_1 \rightarrow e_{1'}$ thus t $\rightarrow$ t', such that t' := e $\leftarrow e_{1'}$.
 \item T-Match - t := match $e_1$ $\overline{\textnormal{case} e_i}$, By the rule T-Match , $\vdash e_1 : State$, by IH, either $e_1$ is a value in which case $\exists. e_j \in \overline{e_i}$ such that \textit{State($e_j$)} $<:$ \textit{State($e_1$)}, and t $\rightarrow$ t', where t' = body of case $e_j$. Else, if $e_1$ $\rightarrow e_{1'}$, t $\rightarrow$ t'', where t'' := match $e_{1'}$ $\overline{\textnormal{case} e_i}$.

 \item T-Case - The argument of T-Case is standard , where the expression is reduced to the body of the case expression.

 \item T-mcall - t := e.m($e_1$, $e_2$,...,$e_p$)- Reduced to cases -
  \begin{itemize}
   \item By IH on the expression e and each of $e_i$ $ 1 \leq i \leq p $, e and $e_i$ is a value, in this case t is reduced to [e/\textit{this} , $e_i$/$x_i$]$e_m$, where \textit{this} is the base object and each of $x_i$ are the formal argument in the method declration and $e_m$ is the body of the method m.

   \item if e is a value and $\exists e_i $ $ 1 \leq i \leq p $, such that $e_i \rightarrow e_{i'}$, then t $\rightarrow$ t' with t' := e.m($e_1$, $e_2$,... $e_{i-1}$, $e_{i'}$...,$e_p$).

   \item if e $\rightarrow$ e' then t $\rightarrow$ t' with t' := e'.m($e_1$, $e_2$,..., $e_p$).
   
  \end{itemize}

\item T-While - t := while [$\exists.\phi$] ($e_1 : Bool$, $e_2$); stmt , this is a standard While case with case wise split for $e_1$ = true and false.
\end{itemize}

\end{proof}

\begin{theorem}[Preservation]
\label{preservation}
if $\Gamma$, $\Phi$ $\vdash$ t : $\tau$ and  t $\rightarrow$ t', then ($\Gamma'$, $\Phi'$) $\vdash$ t' : $\tau'$ and  ($\Gamma'$, $\Phi'$) $\vdash$ $\tau'$ type.
\end{theorem}
\begin{proof}
% 
% \begin{lemma}[Substitution Lemma]
% 
%  The substitution operation preserves type in our type system.
%  
% \end{lemma}

The proof is by the induction on the derivation of $(\Phi, \Gamma) \vdash t : \tau$
% and the substitution lemma for simple and dependent types.
We present the argument about the preservation for an important subset of cases and for others the argument is similar. At each step of the induction we assume that by Induction hypothesis, the preservation lemma holds and then to complete the induction argument we prove that the argument hold for the current step.
\begin{itemize}
 \item T-New, T-New-Dep, T-mDecl, since these are values and thus $\nexists$ t' such that t $\rightarrow$ t' and thus the argument vacuously holds for these typing derivation rules.
 \item T-F-Ref , t := e.f : $\tau$ type, now by IH if e $\rightarrow$ e' then t $\rightarrow$ t', where t' := e'.f and e' is well typed. By T-F-Ref ($\exists \phi_e, s_e, \tau_e \ and \ \phi_{e'}, s_{e'}, \tau_{e'}$), such that $\Gamma(e) := (\phi_e, s_e).\tau_e$ and $\Gamma(e') := (\phi_{e'}, s_{e'}).\tau_{e'}$. Let \textit{sdecl $s_{e'}$ = state $s_{e'}$ case of $s_x$ \{... $f:\tau_{f}$..\}}, thus the type of t' := $\tau_f$.

 \item T-Update, t := e $\leftarrow$ e'  and $(\Phi, \Gamma \vdash)$ t : $\tau$, if e' $\rightarrow$ e'', then t $\rightarrow$ t' and t' := e $\leftarrow$ e''. By IH if $(\Phi, \Gamma) \vdash$ e' : $\tau'$ then after e' $\rightarrow$ e'', $(\Phi, \Gamma) \vdash$ e'' : $\tau''$. Thus by T-update, $(\Phi, \Gamma \vdash)$ t' : $\tau''$). 
  
  \item T-match, t := match $e_1$ $\overline{\textnormal{case} e_i}$, $(\Phi, \Gamma) \vdash$ t : $\tau_1 \rightarrow \tau_u$. There are two possible ways of reduction of t $\rightarrow$ t'- 
    \begin{itemize}
     \item If $e_1 \rightarrow e_{1'}$, then t $\rightarrow$ t', such that t' := match $e_{1'}$ $\overline{\textnormal{case} e_i}$. By IH if $(\Phi, \Gamma) \vdash e_1 : \tau_1$ then $(\Phi, \Gamma) \vdash e_{1'} : \tau_1'$. By T-match, $(\Phi, \Gamma) \vdash t' : (\tau_1' \rightarrow \tau_u)$.
    \item If for some $e_i$, $e_i \rightarrow e_{i'}$ then t $\rightarrow$ t', such that t' := match $e_{1}$ $\overline{\textnormal{case} \ e_{i-1}} \ case e_{i'} \ \overline{\textnormal{case} \ e_{i+1}}$. By IH, if $(\Phi, \Gamma) \vdash e_i : tau_i$ then $(\Phi, \Gamma) \vdash e_{i'} : \tau_i'$. By T-match, let $\tau_u'$ = $\bigcup \tau_{1}...\tau_{i-1} \tau_{i'}..\tau_{k}$ then $(\Phi, \Gamma) \vdash t' : (\tau_1 \rightarrow \tau_u')$.
    \end{itemize}

  \item T-let, t := let x = $e_1$ in e. There are two distinct possibilities of reduction of t $\rightarrow$ t'-
  
  \begin{itemize}
   \item If $e_1 \rightarrow e_{1'}$, by IH $(\Phi, \Gamma) \vdash$ $e_1' : \tau_1'$. Let $(\Phi, \Gamma, x : \tau_1', e_1' : \tau_1') \vdash$ e : $\tau'$, then t' := let x = $e_1'$ in e and $(\Phi, \Gamma) \vdash$ t' : $\tau'$.
  \item If e $\rightarrow$ e', by IH $(\Phi, \Gamma) \vdash e' : \tau'$. thus for t $\rightarrow$ t', $(\Phi, \Gamma) \vdash$ t' : $\tau'$.
  \end{itemize}
% 
%   \item T-case, t := case e \{ $e_b$\}. There are two distinct possibilities of reduction of t $\rightarrow$ t'-
%   \begin{itemize}
%    \item If e $\rightarrow$ e', by IH, $(\Phi, \Gamma) \vdash$ e : $\tau_1'$. If t $\rightarrow$ t' with 
%   \end{itemize}

  \item T-mcall, t := e.m($e_1, e_2, ...,e_p$). There are two distinct possibilities of reduction of t $\rightarrow$ t'-
  \begin{itemize}
   \item e.m(...) $\rightarrow$ e'.m(...), if e $\rightarrow$ e'. By IH, let $(\Phi, \Gamma) \vdash$ e' : $\tau_b'$. By T-mcall, let $((\Phi \wedge (\bigwedge_{i} \Phi_i) (\Gamma, e' : \tau_b', \overline{e_i : \tau_i)}) \vdash e_m : T_r')$ then t' : $T_r'$.
    \item e.m(...,$e_k$,...,$e_p$) $\rightarrow$ e.m(...,$e_k'$,...,$e_p$) for some $k \in [1,p]$ if $e_k \rightarrow e_k'$. By IH $(\Phi, \Gamma) \vdash$ $e_k'$ : $\tau_k'$. By T-mcall, let $((\Phi \wedge (\bigwedge_{i} \Phi_i) (\Gamma, e : \tau_b, \forall i \in \{[1, p]\setminus k\} \overline{e_i : \tau_i)}, e_k' : \tau_k') \vdash e_m : T_r')$, then t' : $T_r'$.

  \end{itemize}

  \item T-while, $\textnormal{while [$\exists. \phi$] ($e_1$) \{e\}}$. Again two distinct possible way of reduction of t $\rightarrow$ t'-
  \begin{itemize}
   \item If $e_1 \rightarrow e_1'$, by T-while $e_1' : bool$, and let $(\Phi_1 \wedge (e_1' == true), (\Gamma, e_1' : bool)) \vdash e : (\Phi_2, \tau') \ \ \Phi_2 \vDash \exists.\phi$, then t' : $\tau'$.
   \item If e $\rightarrow$ e', By IH $(\Phi_1 \wedge (e_1 == true), (\Gamma, e_1 : bool)) \vdash e' : (\Phi_2, \tau') \ \ \Phi_2 \vDash \exists.\phi$, then t' : $\tau'$.
  \end{itemize}

\end{itemize}

\end{proof}

%%%%%%%%%%%%%%%%%%%%%%%%%%%%%%%%%%%%%%%%%%%%%%%%%%%%%%%%%%%%%%%%%%%%%%%%%%%%%%%%%%%%%%%%%%%
\subsection{Proof of Decidability of Typechecking}
\label{decidability}
The typecheking problem for the BR-Typestate, is always reducible to constraint solving over Presburger Arithmetic formulas. Since the Presburger Arithmetic has a decidable and tractable validity problem, this makes the type checking decidable in our typestate system. 

\begin{theorem}[Reduction to PAF]
\label{reduction}
For any general typing relation $(\Phi, \Gamma) \vdash t : (\Phi', \tau)$ in our typestate system, $\exists. \psi \in Presburger Arithmetic Formula$, such that $(\Phi, \Gamma) \vdash t : (\Phi', \tau)$ holds iff $\psi$ is satisfiable.
\end{theorem}
\begin{proof}
The proof is using an inductive argument on the typing derivations for formation, well formedness and subtyping in our typestate system. The routine $\psi(\tau)$ defines the presburger formula for $\tau$.
We consider here only the base types and other complex types and show the PAF $\psi$ for each of these. 
\begin{itemize}
 \item Base case : $\forall$ primary type $\tau \in \{ void, int, bool, String \}$, $\psi(\tau)$ = $\phi_{\tau}$ = $\exists x_{\tau}. x_{\tau} \neq 0$.
  \item Case :: $\tau = S$, let $x_s$ define a variable for the state S, then the formula $\psi(\tau) = x_s \neq 0$.
  \item Case :: $\tau_i <: \tau_j$, by IH let $\psi(\tau_i)$ = $\phi_{\tau_i}$ and $\psi(\tau_j)$ = $\phi_{\tau_j}$, then $\psi(\tau_i <: \tau_j)$ = $\phi_{\tau_i} \vDash \phi_{\tau_j}$.
  \item Case :: $\tau_i = \tau_j$, by IH let $\psi(\tau_i)$ = $\phi_{\tau_i}$ and $\psi(\tau_j)$ = $\phi_{\tau_j}$, then $\psi(\tau_i = \tau_j)$ = $\psi(\tau_i <: \tau_j) \wedge \psi(\tau_j <: \tau_i)$ .
  \item Case :: $\tau_i \rightarrow \tau_j$, By expression typing rules, $\exists. mdecl = \tau_j m ( \tau_i \ a_i) \{... e_b : \tau_b ... \}$. By IH let $\psi(\tau_i)$ = $\phi_{\tau_i}$ and $\psi(\tau_j)$ = $\phi_{\tau_j}$ and $\psi(\tau_b) = \phi_{\tau_b} $, then $\psi(\tau_i \rightarrow \tau_j)$ = $\psi(\psi((\tau_i) \wedge \psi(\tau_b)) <: \psi(\tau_j))$.
  \item Case :: $\tau_i \gg \tau_j$, the case is similar to the $\tau_i \rightarrow \tau_j$ above.

\end{itemize}

\end{proof}

%%%%%%%%%%%%%%%%%%%%%%%%%%%% Train Algorithm %%%%%%%%%%%%%%%%

\subsection{Train Speed-Control Protocol}
he train speed control algorithm controls the speed of the train and guarantees the collision free running of the trains. A train could be in one of the four states \textit{viz. ontime, braking , late or stopped}. Thus a safety property for such a control system could be defined as - \textit{``the train is never late (or early) by more than 20 seconds''}. The speed control system is regulated via counters keeping track of number of beacons \textit{b} passed on the rails and a global clock ticks \textit{s}, besides this there is another counter which starts in the braking state and counts the ticks during breaking state \textit{d}. Each state is defined as - The train is ontime iff $ s-9 < b < s+9$, its late iff $ b \in [s-9, s-1]$, its early iff $b \geq s+9$ finally, when $b=s+1$, the train is on time again.

One property of interest to avoid collisions is- $\forall time, \mid b - s \mid \leq 20$,
which could not be enforced using regular typestate. We present a counter machine for the train speed control protocol in appendix section figure ~\ref{fig:train-speed-control-fig}. 
\begin{figure*}
\begin{tikzpicture}[->,>=stealth',shorten >=1pt,auto,node distance=3cm,
                    semithick]
  \tikzstyle{every state}=[fill=white,draw=black,text=black]

  \node[state, initial] (A)                    {$time$};
  \node[state] 	 (B)  at (6,0)    {$brake$};
  \node[state] 	 (C)  at (6,-4)    {$stop$};
  \node[state] 	 (D)  at (0, -4)    {$late$};

  \path (A) edge [bend left]  node {$G(b=s+9, b'=b+1 \wedge d'=0)$} (B)
        
	(B) edge [bend left]  node [above] {$G(b=s+1, s'=s+1 \wedge d'=0)$} (A)
	(A) edge [loop above]  node [left] {\parbox{3cm}{$G(b<s+9, b'=b+1) \vee G(b>s-9, s'=s+1)$}} (A)
	(B) edge [loop above]  node [right] {\parbox{3cm}{$G(d<9, d'=d+1 \wedge b'=b+1) \vee G(b>s+1, s'=s+1$}} (B)

	(B) edge  node [right] {$G(d=9, b'=b+1)$} (C)
        
	(C) edge [loop below]  node [below] {$G(b>s+1, s'=s+1)$} (C)

	(C) edge   node [right] {\parbox{3cm}{$G(b=s+1, s'=s+1 \wedge d'=0)$}} (A)
	
	(A) edge [bend left]  node [above] {\parbox{3cm}{$G(b=s-9, s'=s+1)$}} (D)
	
	(D) edge [bend left]  node [left] {\parbox{3cm}{$G(b=s-1, b'=b+1)$}} (A)

	(D) edge [loop below]  node [below] {$G(b<s-1, b'=b+1)$} (D);

\end{tikzpicture}
\caption{Counter machine for train speed control system, property $\mid b - s \mid \leq 20$ }
\label{fig:train-speed-control-fig}
\end{figure*}
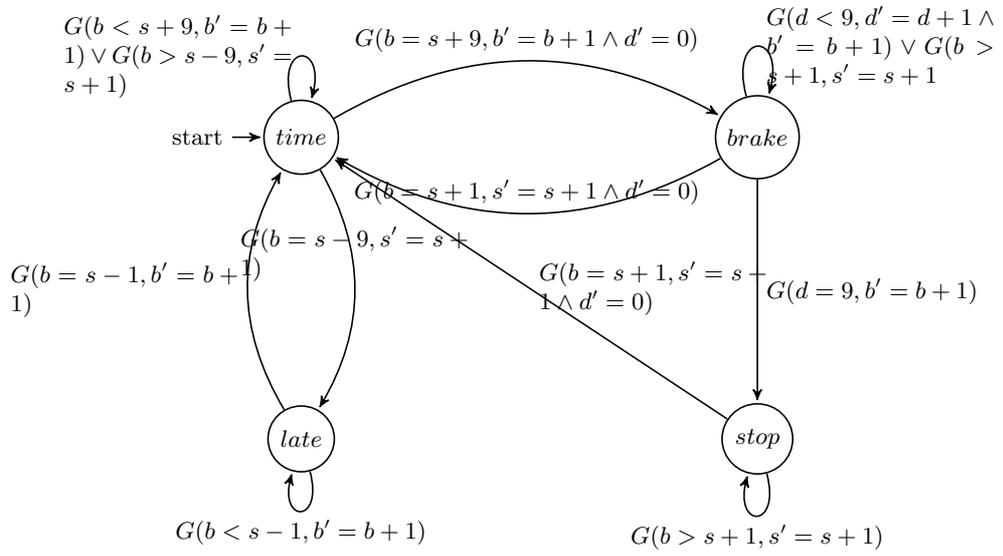